%% file: 0_main.tex
\documentclass[12pt]{article}

\usepackage{graphicx,psfrag,epsf}
\usepackage{enumerate}
\usepackage{natbib}
\usepackage{url} 

\usepackage[a4paper]{geometry}
\usepackage{graphicx}
\usepackage{amsmath,amsfonts,bm,amssymb,amsthm}
\usepackage{bm}
\usepackage{bbm}
\usepackage{array,dcolumn,booktabs,float}
\usepackage{setspace}
\usepackage{verbatim}
\usepackage{color}
\usepackage{soul}
\usepackage{url}
\usepackage{hyperref}
\usepackage{pdflscape}
\usepackage{xcolor}
\usepackage{enumerate}
\usepackage{enumitem}
\usepackage{subfigure}
\usepackage{commath}
\usepackage{mathrsfs}
\usepackage{xfrac} 
\usepackage{multirow}
\usepackage{algpseudocode}
\usepackage[ruled]{algorithm2e}
\usepackage{comment}
\usepackage[normalem]{ulem}

\newcommand{\R}{\mathbb{R}}

\newcommand{\Psib}{\bm{\Psi}}
\renewcommand{\S}{{\cal S}}
\newcommand{\F}{{\cal F}}
\newcommand{\T}{{\cal T}}
\newcommand{\HS}{{Hilbert-Schmidt }}
\renewcommand{\vec}{\mbox{vec}}
\newcommand{\kronecker}{\otimes}

\newcommand{\bpsi}{\boldsymbol{\psi}}

\usepackage{mathtools}
\DeclarePairedDelimiter\ceil{\lceil}{\rceil}

\newtheorem{theorem}{Theorem}

\newtheorem{remark}{Remark}

\newtheorem{proposition}{Proposition}

\begin{document}

\def\spacingset#1{\renewcommand{\baselinestretch}%
{#1}\small\normalsize} \spacingset{1}


\newcommand{\blind}{1}

\if1\blind
{
\title{Locally sparse function-on-function regression}
\author{Mauro Bernardi and Antonio Canale, \\ 
{\small Dipartimento di Scienze Statistiche},\\ {\small Università degli Studi di Padova, Padova, Italy,}\\  
and\\
Marco Stefanucci\\
{\small Dipartimento di Scienze Economiche, Aziendali, Matematiche e Statistiche ``Bruno de Finetti'',}\\ 
{\small Università degli Studi di Trieste, Trieste, Italy}
  }
\date{}
  \maketitle
} \fi

\if0\blind
{
  \bigskip
  \bigskip
  \bigskip
  \begin{center}
    {\LARGE\bf Locally sparse function-on-function regression}
\end{center}
  \medskip
} \fi


\bigskip
\begin{abstract}
In functional data analysis, functional linear regression has attracted significant attention recently. Herein, we consider the case where both the response and covariates are functions. There are two available approaches for addressing such a situation: concurrent and nonconcurrent functional models. In the former, the value of the functional response at a given domain point depends only on the value of the functional regressors evaluated at the same domain point, whereas, in the latter, the functional covariates evaluated at each point of their domain have a non-null effect on  the response at any point of its domain. To balance these two extremes, we propose a locally sparse functional regression model in which the functional regression coefficient is allowed (but not forced) to be exactly zero for a subset of its domain.  This is achieved  using a  suitable basis representation of the functional regression coefficient and exploiting an overlapping group-Lasso penalty for its estimation. We introduce efficient computational strategies based on majorization-minimization algorithms and discuss appealing theoretical properties regarding the model support and consistency of the proposed estimator. We further illustrate the empirical performance of the method through simulations and two applications related to human mortality and bidding the energy market.
\end{abstract}

\noindent%
\vfill

\newpage
\spacingset{1.45} 

\input{1_intro}
\input{2_methodology}
\input{3_computation}

\input{4_theory}
\input{5_simstudy}
\input{6_application}
\input{7_discussion}

\bibliographystyle{rss}
\bibliography{biblio}

\input{8_appendix}

\end{document}

%% file: 1_intro.tex

\section{Introduction}
\label{sec:intro}
The undergoing  technological advancement enables the collection and storage of high-resolution data that can be modelled as smooth functions (e.g., curves or surfaces). Functional data analysis (FDA) is a branch of statistics that models such data through suitable functional counterparts of successful methods and models developed for standard Euclidean data, such as clustering, regression, and classification \citep{FDA1,IFDA,hsing2015theoretical}.
Functional regression is one of the building blocks of FDA and has received remarkable attention in theory, methods, and applications \citep[see][for a recent review]{morris:2015:functionalregression}.

Herein, we focus on the general case of functional regression with functional responses. We developed a model and related methods for its implementation, representing a bridge between two successful alternative models.
To set up the notation, we assume that for the generic $i$-th statistical unit $(i=1, \dots, n)$, a functional response  $y_i(s)$ with $s \in \cal S$ is available along with $p$ functional covariates $x_{ij}(t)$ ($j=1, \dots, p$ and  $t \in {\cal T}_j$) and $\cal S$ and ${\cal T}_j$ $(j =1, \dots, p)$  subsets of $\mathbb R$, denoting the domains of the functional data $y$ and $x_j$, respectively.  The first available modelling strategy, namely, the concurrent functional linear model, assumes that the functional data are observed on the same domain, i.e., ${\cal T}_j = \cal S$  for $j =1, \dots, p$, and the relation between the response and the predictors is given as
\begin{equation}
y_i(s) = \alpha(s)+\sum_{j=1}^p x_{ij}(s)\psi_j(s)+ e_i(s), 
\label{concurrent}
\end{equation}
where $\alpha(s)$ is a functional intercept, $\psi_j(s)$ are functional regression coefficients, and $e_i(s)$ are functional zero-mean random errors. In the concurrent model, the covariates $x_j$ influence $y(s)$ only through their values $x_j(s)$ at the domain point $s \in \cal S$. As a more general approach, the nonconcurrent functional linear model allows $y_i(s)$ to entirely depend on the functional regressors, and specifically,
\begin{equation}
y_i(s) = \alpha(s)+  \sum_{j=1}^p \int x_{ij}(t)\psi_j(t,s) dt + e_i(s),
\label{nonconcurrent1}
\end{equation}
where $\psi_j(t,s)$ is the kernel function determining the impact of  $x_{ij}$ evaluated at domain point $t \in  {\cal T}_j$ on $y_i(s)$. 
The nonconcurrent model is the default choice when the response variable and the covariates do not share the same domain or, even when sharing the same domain, the value of $y_i(s)$ can be assumed to depend on the functional regressors entirely and not just for their values at $s\in \S$. The nonconcurrent model offers great flexibility, but the flexibility increases complexity, both in terms of interpretation and computation. 
An in-between solution, motivated by applications in which ${\cal T}_j = {\cal T} = {\cal S}$ is a time domain, is represented by the so-called historical functional linear model of \citet{malfait:2003:historical}. This approach restricts the domain of integration of the integral in \eqref{nonconcurrent1} to the set  $\T{(s)} = \{t \in \T: t<s\}$ leading to 
\begin{equation}
y_i(s) = \alpha(s)+  \sum_{j=1}^p \int_{\T{(s)}} x_{ij}(t)\psi_j(t,s) dt + e_i(s).
\label{Malfait}
\end{equation}
Note that this model, for the general point $s \in \S$, can be interpreted as a nonconcurrent model up to point $s$.
Despite providing an interesting intermediate solution, there are many situations in which  the dependence through the interval $\T{(s)}$ is difficult to be justified.

Herein, we introduce a hybrid solution  that combines the simplicity in terms of interpretation of the concurrent and historical models with the flexibility of the nonconcurrent model.  Our goal is to introduce a noncconcurrent functional linear model that allows for local sparsity patterns. Specifically,  we want that $\psi_j(t,s)=0$ for $(t, s)  \in \F_0$ with $\F_0 $ being a suitable subset of $\F$,  thus inducing locally sparse \HS operators $\psi_j$. In defining $\F_0$, we do not consider those regions where $\psi(t,s)$ is zero because the kernel changes its sign or where it is tangential to the plane 0. A formal definition of $\F_0$ used henceforth is thus
\[
\F_0 = \{ (t,s) \in \F: \psi(t,s)=0, 
\exists \bar{N}_{} \subseteq  B_\epsilon((t,s)); \, \mu( \bar{N}_{}  )>0
\, \mbox{and}\, \psi(t',s') = 0, \, \forall \, (t',s') \in \bar{N}, \forall \epsilon >0
\},
\]
where $B_\epsilon(t,s)$ is a ball of radius $\epsilon >0$ of the point $(t,s)$ and $\mu(\cdot)$ is the Lebesgue measure. 
Note that, although in \eqref{Malfait} the region $\F_0$ is fixed \emph{a priori} and set equal to $\F_0 = \{(t,s)\in \F : t \geq s\}$, we do not require any specific sparsity pattern, rather we learn it from the data, as discussed in the following sections. 

Different approaches to sparsity and regularization have been used in regression models for functional data, but none of them can be readily adapted for our purpose since they are all defined in a simpler function on scalar regression setting, i.e.,  $y_i$ are scalar responses. 
%
%
 \citet{lee:park:2012:sparse}, for example, after representing the functional regression coefficient with a splines basis expansion,  introduced a Lasso-type penalty for basis coefficients. The proposed solution has the great benefit of regularizing the estimator of  the functional regression coefficients and enjoys nice asymptotic properties. However, several zeroes in the basis function coefficient do not map to a zero in the functional object represented by the basis, and thus, using this approach would not necessarily induce sparsity in the functional coefficients. 
 \citet{james:etal:2009:functional} proposed a model with an interpretable functional regression coefficient that can be exactly zero, flat, and different from zero, or linear in local areas of its univariate domain. This is achieved by inducing sparsity in the general $D$-th derivative of the functional regression coefficient using a constant basis expansion and an $\ell_1$ penalization.
On the surface, this solution seems intuitive and successful, but a zero on a general subregion of the regression function requires grid points that fall into the region to be simultaneously zero, which the plain $\ell_1$ regularization does not warrant in general. In addition, in our setting, the regression coefficient is not a univariate curve but a bivariate surface, and  the use derivatives in two dimensions is less straightforward to apply. 
\citet{lin2017locally} proposed a smooth and locally sparse estimator of the coefficient function based on the combination of smoothing splines with a functional smoothly clipped absolute deviation  (SCAD) penalty of \citet{fan2001variable}.   \citet{zhou2013functional} proposed a two-stage sparse estimator exploiting the properties of B-spline  basis expansion, where an initial estimate is obtained using a Dantzig selector \citep{dantzig}, and the spline coefficient is later refined using a group adaptation of the SCAD penalty. 
The only contribution dealing with a function-on-function situation is the recent manuscript by \citet{vantinieamici} who estimated sparse functional coefficients  through a functional version of the Lasso penalty.

Our proposed solution exploits a B-spline local property  similarly to  \citet{zhou2013functional}, but directly relies on minimizing a suitable objective function, as discussed in Section~\ref{sec:mod}. This objective function  includes an overlapping group-Lasso penalty \citep{jenatton2011structured}  that  ensures  the desired sparsity in  $\psi_j$. It is minimized using a fast and reliable numerical strategy proposed in Section~\ref{sec:comp}. The properties of the induced estimator are discussed in Section~\ref{sec:properties}. A detailed simulation study to evaluate the empirical performance of the proposed model compared to that of the state-of-the-art models is presented in Section \ref{sec:sim}. The results show that the proposed model outperforms other models in estimating both the region of sparsity and the value of the functional regression coefficient where it is different from zero. In Section~\ref{sec:app}, we specify the model in functional time-series settings and analyze two datasets related to  human mortality and bidding in energy markets. Our analysis suggests that bridging between concurrent and nonconcurrent functional models results in better performance both in terms of goodness-of-fit and qualitative interpretation of the results.
Section \ref{sec:discussion} concludes the paper.

%% file: 2_methodology.tex
\section{Locally sparse functional model}
\label{sec:mod}
%

Without loss of generality, we consider the case in which the functional data $y_i$ are centered in zero, and  a single functional covariate $x_i$ is available for $i =1, \dots, n$, so that $\alpha(s)=0$ for all $s \in \S$ and $p=1$. We also assume that $\psi$ is bounded and defined on a compact domain $\F$. Thus,  \eqref{nonconcurrent1} becomes
\begin{equation}
y_i(s) = \int x_i(t)\psi(t,s) dt + e_{i}(s).
\label{eq:modelnointercept}
\end{equation}
The proposed locally sparse functional regression (LSFR) model relies on the introduction of a specific basis representation for the kernel $\psi$ and the minimization of a suitable objective function. These are described in the next two sections.

\subsection{Kernel basis representation}
\label{sec:kerbasis}
We employ the common  FDA  practice of representing the functional objects using of basis expansions. Specifically, we select two bases in $L_2$, e.g., $\{\theta_l(s), l=1,\dots,L\}$ and $\{\varphi_m(t), m=1,\dots,M\}$, 
where each $\theta_l$ is defined on $\cal S$, each $\varphi_m$ is defined on $\cal T$, and the number of basis is  $L$ and $M$, respectively. Exploiting a tensor product expansion of these two, we represent the kernel $\psi$ in \eqref{eq:modelnointercept} as
\begin{equation*}
\psi(t,s) = \sum_{m=1}^M\sum_{l=1}^L\psi_{ml} \varphi_m(t) \theta_l(s),
\end{equation*}
alternatively, in matrix form as
\begin{eqnarray}
\psi(t,s) & = & (\varphi_1(t), \dots,  \varphi_M(t)) 
\left(
\begin{matrix}
\psi_{11} & \cdots & \psi_{1L} \\
\vdots & \ddots& \vdots \\
\psi_{M1} & \cdots & \psi_{ML} 
\end{matrix}\right)
\left(
\begin{matrix}
\theta_{1}(s)\\
\vdots \\
\theta_{L}(s)
\end{matrix}\right) =
{\bm \varphi}(t)^T  \bm{\Psi} {\bm \theta}(s).
\label{eq:tensorproduct}
\end{eqnarray}
where $\psi_{ml} \in \R$ for $l=1,\dots,L$ and $m=1,\dots,M$.

To achieve the desired sparsity property in $\psi(t,s)$, similarly to \citet{zhou2013functional}, first assume that  the elements in equation \eqref{eq:tensorproduct} are  B-splines \citep{deboor} of order $d$. A B-spline of order $d$ is a piecewise polynomial function of degree $d-1$ and is defined by a set of knots, which represent the values of the domain where the polynomials meet. Based on \eqref{eq:tensorproduct}, $\{\theta_l(s), l=1,\dots,L\}$ and $\{\varphi_m(t), m=1,\dots,M\}$ are  B-splines of order $d$ with $L-d$ and $M-d$ interior knots, respectively, and two external knots each.  

Suitable zero patterns in the B-spline basis coefficients of $\Psib$ induce sparsity of $\psi(t,s)$. Let $\tau_1 < \dots< \tau_m< \dots  < \tau_{M-d+2}$ and $\sigma_1 < \dots <\sigma_l < \dots < \sigma_{L-d+2}$ denote the knots defining the tensor product splines in \eqref{eq:tensorproduct}, with $\tau_m \in \T$ and $\sigma_l \in \S$. For  $m=1, \dots, M-d+1$ and $l=1,\dots, L-d+1$ let $\F_{m,l} \in \F$ be the rectangular subset of $\F$ defined as
\begin{equation}
\F_{m,l} =  (\tau_{m}, \tau_{m+1}) \times (\sigma_{l}, \sigma_{l+1}).
\label{eq:nullregion}
    \end{equation}
Hence, to obtain $\psi(t,s) = 0$ for each $(t,s) \in \F_{m,l}$, it is sufficient that all the coefficients $\psi_{m',l'}$ with  $m'=m, \dots, m+d-1$ and $l'=l, \dots,l+d-1$ are  jointly zero.
 In general  $\psi(t,s)$ equals zero in the region identified by two pairs of consecutive knots if the related  $d \times d$ block of coefficients of $\Psib$ is entirely set to zero.
This  suggests that $\Psib$ should be suitably partitioned in several blocks of dimensions $d \times d$ on which a joint sparsity penalty is induced. This is discussed further in the next section.

\subsection{Sparsity-inducing norm}

Consistent with the above discussion and with successful approaches in statistics and machine learning,  we  minimize an objective function having the following form:
\begin{equation}
\frac{1}{2}\sum_{i=1}^n \int \left( y_i(s) - \int x_i(t) \psi(t,s) dt \right)^2 ds + \lambda \Omega(\Psib),
\label{eq:min1}
\end{equation}
where the first summand is a goodness-of-fit index while the second a suitable penalty. 

Studies on statistics and machine learning have proposed numerous approaches for the sparsity-inducing  $\Omega$, including  Lasso \citep{tibshirani1996,efron2004least}.
 In our setting, however,  Lasso would yield sparsity by treating each parameter individually, regardless of its position in $\Psib$, which is against our desiderata.
The first extension of Lasso involving the concept of groups of coefficients is the group-Lasso  reported by \citet{yuan:2006:glasso}. This approach considers a partition of all the coefficients into a certain number of (disjoint) subsets and eventually allows some of these groups to shrink to zero.

 \begin{figure}
	\centering
	\includegraphics[width=0.75\textwidth,]{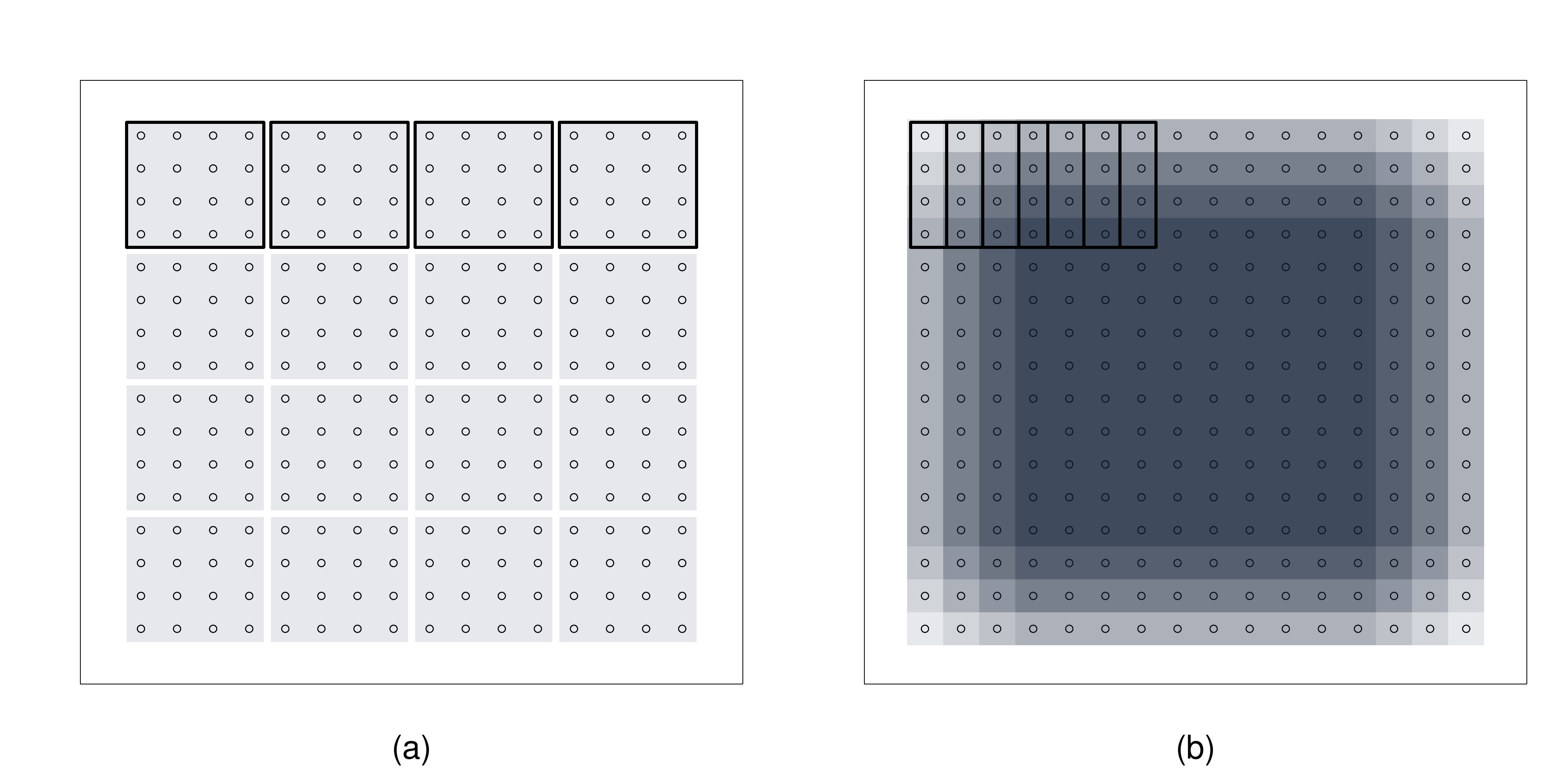}
	\caption{Possible coverings of the matrix $\Psi$ ($L=M=16$) with blocks of size $d\times d$ with $d=4$: disjoint covering (a) and overlap covering (b). Darker color denotes the superimposition of different blocks. The first four blocks are represented with thick black borders in each subfigure.}
	\label{fig:tassellation}
\end{figure}

On the surface, the group-Lasso approach is a promising solution to our problem, but we discuss this in more detail in what follows considering the graphical representation of a possible grouping, coherent with the group-Lasso definition, reported in panel (a) of Figure \ref{fig:tassellation}. Setting all coefficients belonging to one of the tiles of panel (a) to zero will force $\psi(t,s) = 0$ for $(t,s)$ in a specific set  $\F_{m,l}$. In fact, this kind of solution does not provide enough model flexibility in many respects. 
 First, the construction requires $L$ and $M$ to be multiple of $d$.      
 Second, it is impossible to induce $\psi(t,s) = 0$ for $(t,s) \in \F_{m,l}$ and $\psi(t',s') \neq 0$ for $(t',s')$ belonging to neighboring sets of $\F_{m,l}$ for any $m$ and $l$ not equal to $k_1 d+1$ and $k_2 d+1$, respectively, where  $k_1$ and $k_2$ are nonnegative integers. 
 Third, $\psi(t,s) = 0$ for $(t,s)$ belonging to the union of two contiguous rectangles $\F_{m,l} \cup \F_{m+1,l}$ if and only if two consecutive and disjoint blocks of coefficients are jointly zero, leading to $\psi(t,s) = 0$ for each $(t,s)$  belonging to the superset $\bigcup_{j=0}^{d+1} \F_{m+j,l}$. 
  Figure \ref{fig:esempi} in the Appendix shows two examples of these key limitations of the disjoint grouping of coefficients.

 To overcome these issues, instead of a disjoint partition, we define an overlapping sequence of blocks of size $d \times d$. Specifically, we introduce the block index $b = 1, \dots, B$ {with $B= (M-d+1) \times (L-d+1)$} denoting the total number of blocks. {Notably, there is a block for each set $\F_{m,l}$.} The generic $b$-th block contains the coefficients in the set
 \begin{eqnarray}
 \big\{ \psi_{ml}&:&  m=m^*, \dots,  m^*+d-1, \notag\\
  &&l=l^*, \dots, l^*+d-1, m^*= \ceil*{\frac{b}{L-d+1}} , 
 l^*= b \text{ mod } (L-d+1) \big\},
 \label{eq:blocks}
 \end{eqnarray}
  where $A \text{ mod } B$ represents the remainder of the division $A/B$. A graphical representation of this construction is shown in panel (b)  of Figure~\ref{fig:tassellation}.  This overlapping group structure allows $\F_0$ to be the union of any set $\F_{m,l}$ in \eqref{eq:nullregion} by moving a block of minimum size, depending on $d$, with steps of size one, and thus, allowing for greater flexibility in the definition of the subset $\F_0$, with respect to the disjoint grouping of panel (a). A precise characterization is formalized in Proposition 1,  in the next section. 

The above construction suggests specifying a penalty $\Omega$ for overlapping groups of coefficients, which has attracted significant interest in the last decade. For example, \citet{zhao2009composite} focused on overlapping and nested groups of coefficients motivated by modeling hierarchical relationships between  predictors. More general overlapping group-Lasso penalties have been proposed by \citet{jacob} and \citet{jenatton2011structured}, which define suitable norms for inducing a penalty that can model specific patterns for the \emph{support} of the vector of coefficients---the latter being the set of coefficients different from zero. 
%
The difference between these two approaches is that  \citet{jenatton2011structured} introduced a penalty inducing supports that arise as intersections of a subset of suitable groups, whereas \citet{jacob}  introduced a penalty that induces supports that are rather the  unions of a subset of the groups.
\citet{lim:hastie} exploited the construction of \citet{jacob} to learn the main and pairwise interaction terms of categorical covariates in linear and logistic regression models, imposing through the overlapping group structure hierarchical regularization with a similar motivation as \citet{zhao2009composite} but also proposing an efficient computational strategy.

In classical regression, the interest lies in the support of the vector of regression coefficients. Consistent with our motivations, instead, we focus  on the sparsity structure of the matrix of coefficients $\Psib$  rather than its support.
Hence, we specify \eqref{eq:min1} as
\begin{equation}	
\frac{1}{2} \sum_{i=1}^n \int \left( y_i(s) - \int x_i(t) \psi(t,s) dt \right)^2 ds + \lambda \sum_{b=1}^{B+1}  || c_{b} \odot \boldsymbol{\psi} ||_2 ,
\label{eq:min2}
\end{equation}
where $\lambda>0$ is a fixed penalization term, and $\Omega$ in \eqref{eq:min1} specifies in the sum of $B+1$ Euclidean norms  $||c_{b}\odot\boldsymbol{\psi}||_2$, where $\boldsymbol{\psi} = \vec (\Psib)$, and  $\odot$ represents the Hadamard product. The index $b$ denotes the block of coefficients in $\Psib$, with the first $B$ blocks being consistent with \eqref{eq:blocks} and the last block containing all coefficients in $\Psib$. Vectors of size $ML$, denoted by $c_b$, contain known constants that equally balance the penalization of the coefficients in $\bm{\Psi}$. This \emph{balancing} is needed to account for the fact that the parameters close to the boundaries of the matrix $\Psib$ appear in fewer  groups than central ones,  as shown by the color scaling of panel (b) of Figure \ref{fig:tassellation}. Specifically, the generic vector $c_b$ is defined as 
$
c_b = \vec(  \mathbf{S}_b \odot \mathbf{C}),
$ 
where $  \mathbf{S}_b$ is the $M \times L$ selection matrix with general entry $s_{ml}^{(b)}$, defined as 1 if the parameter $\psi_{ml}$ belongs to group $b$ and 0 otherwise and $\mathbf{C}$ is the matrix with general element $c_{ml}$ defined as 
$
c_{ml} = \left( \sum_{b=1}^{B+1}  s_{ml}^{(b)} \right)^{-1}.
$
Note that this  penalty constitutes a special case of the norm defined by \citet{jenatton2011structured}. 



%% file: 3_computation.tex
%
\section{Computation}
\label{sec:comp}
%
Before describing an efficient computational strategy for our LSFR model, we introduce the empirical counterparts of the quantities described in the previous section assuming to observe a sample of  response curves $y_i$ with $i=1, \dots,n$ on a common grid of $G$ points, i.e. $y_i = (y_i(s_1), \dots, y_i(s_G))^T$. Let also  $x_i$ be  the related functional covariate  observed on a possibly different but--- common across $i$---grid of points, that for simplicity and without loss of generality, we assume of length $G$. 
Let $\mathbf{X}$ be the $n \times G$ matrix with  $x_i$ in the rows. Let $\boldsymbol{\Phi}$ and $\boldsymbol{\Theta}$ be the $M \times G$  and $L \times G$ matrices defined as
\[
\boldsymbol{\Phi}  = \begin{pmatrix}
\varphi_1(t_1) &  \cdots & \varphi_1(t_G) \\
\vdots&&\vdots \\
\varphi_m(t_1) &  \cdots & \varphi_m(t_G) \\
\vdots&&\vdots \\
\varphi_M(t_1) &  \cdots & \varphi_M(t_G) \\
\end{pmatrix}, \quad \quad
\boldsymbol{\Theta} = \begin{pmatrix}
\theta_1(s_1) &  \dots & \theta_1(s_G) \\
\vdots&&\vdots \\
\theta_l(s_1) &  \cdots & \theta_l(s_G) \\
\vdots&&\vdots \\
\theta_L(s_1) &  \cdots & \theta_L(s_G) \\
\end{pmatrix}.
\]
Let $\mathbf{Y}$  and ${\bf E}$  be the $n\times G$ matrices obtained as $\mathbf{Y} = (y_1, \dots, y_n)^T$ and  $\mathbf{E} = (e_1, \dots, e_n)^T$ , with  $e_i = (e_i(s_1), \dots, e_i(s_G))^T$. Model \eqref{eq:modelnointercept} can be equivalently written in matrix form as
\[
\mathbf{Y} = \mathbf{X} \boldsymbol{\Phi}^T \Psib \boldsymbol{\Theta}   + {\bf E}.
\]
Applying the vectorization operator on each side of the equality above,  we have
\begin{align*}
	\mathbf{y}  &= \vec(\mathbf{X} \boldsymbol{\Phi}^T \Psib \boldsymbol{\Theta}) + \vec({\bf E})
	= ( \boldsymbol{\Theta}^T  \kronecker \mathbf{X} \boldsymbol{\Phi}^T) \vec(\Psib)+ \vec({\bf E})
	 = \mathbf{Z}\boldsymbol{\psi} + {\bf e},
\end{align*}
where $\mathbf{y}=\vec(\mathbf{Y})$, $\boldsymbol{\psi}=\vec(\Psib)$ is the vector of coefficients of dimension $LM$,
$\bf{e}=\vec(\bf {E})$,
and  $\mathbf{Z}= \boldsymbol{\Theta}^T  \kronecker \mathbf{X} \boldsymbol{\Phi}^T $ is the design matrix of dimension $nG \times LM$. Therefore, for a given tuning parameter $\lambda>0$, the optimization problem becomes  
%
%
\begin{equation}
\begin{aligned}
\widehat{\boldsymbol{\psi}}_\lambda&=\arg \min_{\boldsymbol{\psi}} \ell(\boldsymbol{\psi}), \quad 
\ell(\boldsymbol{\psi})&=\frac{1}{2}\Vert \mathbf{y} - \mathbf{Z} \boldsymbol{\psi}\Vert_2^2 + \lambda \sum_{b=1}^{B+1} \Vert \mathbf{D}_{b}\boldsymbol{\psi} \Vert_2,
\end{aligned}
\label{eq:min3_init}
\end{equation}
where $\mathbf{D}_b=\mathrm{diag}(c_b)$ is a diagonal matrix whose elements correspond to the elements of the vector $c_b$ defined in the previous section.
The following result is about the uniqueness of the solution of the optimization problem in equation \eqref{eq:min3_init} thus leading to the uniqueness of the estimator $\widehat \psi_\lambda$. Its proof is reported in the Appendix.
\begin{theorem}
\label{th:existence}
Under the representation  \eqref{eq:tensorproduct}, for any $\lambda >0$, the minimization of \eqref{eq:min3_init} with respect to the coefficients $\psi_{ml}$ for $l=1, \dots, L$ and $m=1, \dots, M$ leads to a unique solution $\widehat \psi_\lambda$.
\end{theorem}
In practice, however,  the non-separability of the penalty function when groups overlap, makes the  optimization problem in equation \eqref{eq:min3_init} not straightforward. { The non-separability of the overlap group-Lasso penalty function prevents the application of standard coordinate descent algorithms that cycle through the parameters and updates them either individually (as for the Lasso) or by groups (as for the group-Lasso), \citep[see, e.g.][]{wu_etal.2008,bach_etal.2012,huang_etal.2012,yang_zou.2015}}.
%
We propose to map the optimization in equation \eqref{eq:min3_init} to an optimization of a fully convex and differentiable function by leveraging the  Majorization-Minimization (MM, hereafter) principle firstly introduced by \cite{ortega_rheinboldt.1970} and developed by \cite{hunter_lange.2004, lange.2010,  lange.2016}. The MM approach is a general  prescription for constructing optimization algorithms that operates by creating a surrogate function that minorizes (or majorizes) the objective function. The surrogate function is then maximized (or minimized) in place of the original function. 
Various existing approaches in statistics and machine learning can be interpreted from the majorization-minorization point of view, including the EM algorithm \citep[][]{neal_hinton.1999,wu_lange.2010} or boosting and some variants of the variational Bayes methods \citep[][]{wainwright_jordan.2008}. %
Here the MM approach is employed for the purpose of delivering a quadratic function that majorizes the convex objective function in  \eqref{eq:min3_init}. 
The MM algorithm has been introduced within the context of $\ell_1$ and $\ell_2$ penalized regressions by \cite{wu_etal.2008}. The authors consider both the Lasso and group-Lasso penalty and provide cyclic coordinate descent algorithms for both problems. Coordinate descent algorithms are simple, fast and stable and they usually do not require the inversion of large matrices \citep{hunter_lange.2004}. However, for large dimensional models with high level of sparsity and separable penalty, the coordinate-wise gradient provide all the relevant information to update parameters and most of the parameters are never updated from their starting value of zero. Nevertheless, when the level of sparsity is either unknown or moderately low, coordinate-wise updates no longer represent the best strategy. This is exactly the framework here considered. Therefore, leveraging the supporting hyperplane paradigm \citep[see][]{lange.2016} and the Sherman–Morrison–Woodbury identity we deliver an efficient MM algorithm for our overlap group-Lasso penalty that jointly updates the non-zero regression coefficients at any iteration. Observe also that, because of the analytical form of the overlap penalty  function \citep[see][]{huang_etal.2012}, efficient block-wise updates as in \cite{qin_etal.2013} cannot be considered here. {Our MM algorithm instead serves the purpose of delivering an efficient solution to the otpimization problem in equation \eqref{eq:min3_init} without imposing any specific group conformation and finds application even in more structured group-type penalties, \citep[e.g.][]{bach_etal.2012,jenatton2011structured}}.\par
%
%
Any MM algorithm  iterates between two steps: $(i)$ at   iteration $k$, a majorizer function $\mathcal{Q} (\boldsymbol{\psi}\vert\widehat{\boldsymbol{\psi}}^{(k-1)})$ is obtained conditioning on $\widehat{\boldsymbol{\psi}}^{(k-1)}$; $(ii)$ the value for $\widehat{\boldsymbol{\psi}}^{(k)}$ is obtained minimizing  $\mathcal{Q} (\boldsymbol{\psi}\vert\widehat{\boldsymbol{\psi}}^{(k-1)})$. As a further benefit of exploiting the MM majorization principle, the MM iterations $\widehat{\boldsymbol{\psi}}^{(k)}$, for $k=1,2,\dots$ possesses the decent properties driving the target function downhill.
Several paradigms can be exploited to derive a valid majorizer of the overlap group penalty term in  \eqref{eq:min3_init}, \citep[see, e.g.][for an exhaustive discussion]{lange.2016}. Here, we leverage the dominating hyperplane principle that introduces an upper bound for the strictly concave function $\sqrt{x}$, i.e. $\sqrt{x}\leq \sqrt{x^{(k)}}+(2\sqrt{x^{(k)}})^{-1}\big(x-x^{(k)}\big)$ for any $x,x^{(k)}\in\mathbb{R}^+$. This manoeuvres separate parameters and reduce the surrogate to a sum of linear terms and squared Euclidean norms. 
Consistently with this, we introduce the following surrogate function that majorizes \eqref{eq:min3_init} at $\widehat{\boldsymbol{\psi}}^{(k)}$
\begin{equation}
\mathcal{Q}(\boldsymbol{\psi}\vert\widehat{\boldsymbol{\psi}}^{(k)})=\frac{1}{2} \Vert \mathbf{y}  - \mathbf{Z} \boldsymbol{\psi}\Vert_2^2 + \lambda  \sum_{b=1}^{B+1} \biggl( \Vert \mathbf{D}_b\widehat{\boldsymbol{\psi}}^{(k)}\Vert_{2}+\frac{\Vert \mathbf{D}_b\boldsymbol{\psi}\Vert^2_{2} - \Vert \mathbf{D}_b\widehat{\boldsymbol{\psi}}^{(k)}\Vert_{2}^2}{2\Vert \mathbf{D}_b\widehat{\boldsymbol{\psi}}^{(k)}\Vert_{2}}\biggl),
\label{eq15}
\end{equation}
where $\widehat{\boldsymbol{\psi}}^{(k)}$ is the value of the parameter $\boldsymbol{\psi}$ at the $k$-th iteration  of the MM algorithm. The important consequence of this result is that solution of problem \eqref{eq:min3_init} can be obtained through the iterative minimization of (\ref{eq15}), where (\ref{eq15}) is convex and differentiable. A compact form for the minimization problem of the surrogate function of \eqref{eq15} is 
\begin{align}
\label{eq:minimization_surrogate_fun_1}
\widehat{\boldsymbol{\psi}}^{(k+1)}_\lambda &= \arg \min_{\boldsymbol{\psi}}\mathcal{Q}(\boldsymbol{\psi}\vert\widehat{\boldsymbol{\psi}}^{(k)})\\
\label{eq:minimization_surrogate_fun_2}
\mathcal{Q}(\boldsymbol{\psi}\vert\widehat{\boldsymbol{\psi}}^{k})&=\frac{1}{2} \Vert \mathbf{y}  - \mathbf{Z}\boldsymbol{\psi}\Vert_2^2 + \widehat{d}_0^{(k)} + \lambda  \sum_{b=1}^{B+1} \widehat{d}_b^{(k)} \Vert  \mathbf{D}_b\boldsymbol{\psi} \Vert^2_{2},
\end{align}
where $\widehat{d}_0^{(k)}$ and $\widehat{d}_b^{(k)}$ are constants that depend on the $k$-th iteration:
\begin{equation}
\label{eq:mm_algo_update_constants}
\widehat{d}_0^{(k)} = \sum_{b=1}^{B+1} \biggl( \Vert \mathbf{D}_b \widehat{\boldsymbol{\psi}}^{(k)}\Vert_{2} - \frac{\Vert  \mathbf{D}_b\widehat{\boldsymbol{\psi}}^{(k)}\Vert^2_{2}}{2\Vert  \mathbf{D}_b\widehat{\boldsymbol{\psi}}^{(k)}\Vert_{2}} \biggl), \quad\mbox{and} \quad \widehat{d}_b^{(k)} = \frac{1}{2\Vert \mathbf{D}_b\widehat{\boldsymbol{\psi}}^{(k)}\Vert_{2}}.
\end{equation}
The explicit solution of the minimization problem in equations \eqref{eq:minimization_surrogate_fun_1}-\eqref{eq:minimization_surrogate_fun_2} at the $k$-th iteration is
\begin{equation}
\label{eq:minimization_surrogate_fun_sol}
\widehat{\boldsymbol{\psi}}^{(k+1)} = (\mathbf{Z}^{T}\mathbf{Z} + \lambda \mathbf{H}^{(k)} )^{-1} \mathbf{Z}^{T} \mathbf{y}, 
\end{equation}
where $\mathbf{H}^{(k)}=\big(\sum_{b=1}^{B+1} \widehat{d}_b^{(k)}\mathbf{D}_b^T\mathbf{D}_b\big)^{1/2}$. The MM algorithm is described in Algorithm \ref{alg:coord_desc_mm_qr}, reported in the Appendix. An important consideration in using a particular algorithm is the amount of work the computer is required to carry out in running it which is measured in terms of the number of floating-point operations (FLOPS) that are needed. The computational cost of Algorithm \ref{alg:coord_desc_mm_qr} is provided in the Appendix, at Proposition \ref{prop:cost}. 
%

%
%
\subsection{Efficient MM for the generalised ridge inversion}
%
 The MM parameters update in equation \eqref{eq:minimization_surrogate_fun_sol} suffers from two major drawbacks. First, it requires the inversion of a potentially large-dimensional ridge-type design matrix and, second, it does not exclude the pathological case where one or more of the denominators of the weights in  \eqref{eq:mm_algo_update_constants} are exactly zero. The usual solution for the latter problem consists to perturb $\widehat{d}_b^{(k)}$ by adding a small $\epsilon>0$ to the $\ell_2$-norm, \citep[see][]{hunter_lange.2000}. In what follows, instead, we rely on a different solution. 
Specifically, let 
$\mathbf{H}^{(k)}$ be a $LM\times LM$  diagonal matrix and define the $LM\times LM$ symmetric and positive definite matrix $\mathbf{A}_\lambda^{(k)}= \mathbf{Z} ^{{T}} \mathbf{Z} +\lambda\mathbf{H}^{(k)}\in\mathbb{S}^{LM}_{++}$. The computation of the MM update in equation \eqref{eq:minimization_surrogate_fun_sol} requires the inversion of the $LM\times LM$ symmetric full matrix $\mathbf{A}_\lambda^{(k)}$ at any iteration which takes on the order of $\mathcal{O}((LM)^3)$ arithmetic operations \citep[][]{golub_van_loan.2013}. Since the matrix $\mathbf{H}^{(k)}$ changes at any iteration performing the QR decomposition of $\mathbf{A}_\lambda^{(k)}$ becomes prohibitive even for moderately large values of $LM $.\par
%
Given the diagonal structure of the matrix $\mathbf{H}^{(k)}$, using the Sherman-Morrison-Woodbury matrix identity reduces the problem of inverting $\mathbf{A}_\lambda^{(k)}$ for a fixed $\lambda$ to the simpler problem of computing 
\begin{equation}
\big( \mathbf{Z} ^{T} \mathbf{Z} +\lambda\mathbf{H}^{(k)}\big)^{-1}=\frac{1}{\lambda}\big(\mathbf{H}^{(k)}\big)^{-1}-\frac{1}{\lambda^2}\big(\mathbf{H}^{(k)}\big)^{-1} \mathbf{Z} ^{T} \mathbf{B}^{(k)}_{\lambda} \mathbf{Z} \big(\mathbf{H}^{(k)}\big)^{-1},
\end{equation}
where the full matrix $\mathbf{B}^{(k)}_{\lambda}=\lambda\big(\lambda \mathbf{I}_{nG}+\mathbf{J}^{(k)}\big)^{-1}\in\mathbb{S}_{++}^{nG}$ with $\mathbf{J}^{(k)}= \mathbf{Z} \big(\mathbf{H}^{(k)}\big)^{-1} \mathbf{Z} ^{T}\in\mathbb{S}_{++}^{nG}$ 
can be computed without loss of generality  through the spectral decomposition of $\mathbf{J}^{(k)}$. Specifically, let $ \mathbf{U}^{(k)} \boldsymbol{\Lambda}^{(k)} (\mathbf{U}^{(k)})^{T}=\mathbf{J}^{(k)}$ be 
such a decomposition, then 
$\mathbf{B}^{(k)}_{\lambda}=\lambda \mathbf{U}^{(k)}\big(\lambda \mathbf{I}_{nG}+\boldsymbol{\Lambda}^{(k)}\big)^{-1}(\mathbf{U}^{(k)})^{T}$.
Therefore, the MM  update only requires the spectral decomposition of the symmetric matrix $\mathbf{J}^{(k)}$ to be computed at any iteration. However, as pointed by the following remark, as a byproduct of our procedure, we obtain the indirect solution to the problem of zeros in the surrogate penalty function.
\begin{remark}
As iterations proceeds, it may happen that some of the weights $\Vert \mathbf{D}_b\widehat{\boldsymbol{\psi}}^{(k)}\Vert_2$  associated to a sequence of zeros on the vector $\widehat{\boldsymbol{\psi}}^{(k)}$ larger than $d^2$ becomes closer and closer to zero. Leveraging the Sherman-Morrison-Woodbury matrix identity prevents the weights to explode.
\end{remark}
We can exploit the fact that $\widehat{\boldsymbol{\psi}}^{(k)}$ in equation \eqref{eq:minimization_surrogate_fun_sol} at some iteration $k\geq k_0$ may become zero, i.e. $\Vert \mathbf{D}_b \widehat{\boldsymbol{\psi}}^{(k)}\Vert_2=0$, to provide a fast and efficient solution to the problem of finding the spectral decomposition of $\mathbf{J}^{(k)}$ 
for $nG$ moderately large. Let $ \mathbf{Z}_0\in\mathbb{R}^{nG\times p_0}$ and $ \mathbf{Z}_1\in\mathbb{R}^{nG\times (LM-p_0)}$ with $1< p_0< LM$ be a disjoint partition of the space spanned by $ \mathbf{Z} $ such that $ \mathbf{Z} =(\begin{matrix} \mathbf{Z}_0& \mathbf{Z} _1\end{matrix})$, and let $\mathbf{H}^{(k)}_0\in\mathbb{R}^{p_0\times p_0}$ and $\mathbf{H}^{(k)}_1\in\mathbb{R}^{(LM-p_0)\times (LM-p_0)}$ be the corresponding partition of the diagonal matrix $\mathbf{H}^{(k)}$, then 
$$ 
\mathbf{Z} \big(\mathbf{H}^{(k)}\big)^{-1} \mathbf{Z} ^{T}= \mathbf{Z}_0\big(\mathbf{H}_0^{(k)}\big)^{-1} \mathbf{Z}_0^{T}+ \mathbf{Z}_1\big(\mathbf{H}_1^{(k)}\big)^{-1} \mathbf{Z}^{T}_1.
$$ 
Now, assume that at the $k$-th iteration, $(\mathbf{H}_0^{(k)})^{-1}=\mathrm{diag}\{0,\dots,0\}$, then $ \mathbf{Z} (\mathbf{H}^{(k)})^{-1} \mathbf{Z} ^{T}= \mathbf{Z}_1(\mathbf{H}_1^{(k)})^{-1} \mathbf{Z} ^{T}_1$, which only requires the spectral decomposition of the matrix $\mathbf{Z}_1(\mathbf{H}_1^{(k)})^{-1} \mathbf{Z} ^{T}_1$.
Additional computational considerations and results are reported in the Appendix.

%% file: 4_theory.tex
\section{Theoretical properties}
\label{sec:properties}

In this section, we provide appealing theoretical properties for our LSFR model and for the related estimator $\widehat \psi_\lambda(t,s)$  of $\psi(t,s)$ arising from \eqref{eq:tensorproduct} and \eqref{eq:min3_init}. All the proofs are reported in the Appendix. We first describe the sparsity patterns that we can induce through our LSFR construction in the following proposition.
\begin{proposition}
	\label{prop:F0}
	Let $F = \{\F_{m,l}; \, m=1, \dots, M-d+1; \, l =1, \dots, L-d+1\}$  be the collection of all the sets described in \eqref{eq:nullregion}, ${\cal{P}}(F)$ the power set of $F$, and ${\cal F}_0(\psi)$ be the subset of $\F$ in which $\psi(t,s) = 0$.	Then $\widehat{\psi}_\lambda \in S$ with $S  \subseteq L_2$ and
	$
	S = \{ \psi \in L_2:  \F_0(\psi)  \in {\cal P}(F)\}.
	$
\end{proposition}
The following two results show that the introduced model structure is sufficiently flexible  to cover the two extreme situations represented by the concurrent and nonconcurrent models. 
\begin{remark}
	\label{rem:concurrentseidentro}
	Assume $L = M$. Let $\F_{0}^\text{concurrent} = \{(t,s) \in \F: t\neq s\}$ be the subset of $\F$ where $\psi(t,s)=0$ in the concurrent model. Let $\bar{ \F}_0 \in \mathcal{P}(F)$ the superset of $\F_0^\text{concurrent}$ defined as 
	\[
	{\bar \F}_0 = \arg\,\min_{\F_0 \in \mathcal{P}(F)} \mu (\F_0 \backslash   \F_0^\text{concurrent}). 
	\]
	Then $\mu( \bar \F_0 \backslash \F_0^\text{concurrent}) = O(1/M)$. 
\end{remark}
\begin{remark}
\label{rem:fullseidentro}
	If $\psi^*(t,s) \neq0$ for each $(t,s) \in \cal F$ then $\psi^* \in S$.
\end{remark}
Remark \ref{rem:concurrentseidentro} implies that for a sufficiently high number of knots in the tensor product spline expansion, the sparsity structure of a concurrent Hilbert-Schmidt operator can be well approximated by the proposed formulation. A similar property holds also for the historical model \citep{malfait:2003:historical}. At the same time, Remark \ref{rem:fullseidentro} trivially states that  nonconcurrent model is included in our class, implying that  relatively small values of the penalty parameter $\lambda$ lead to no sparsity of $\psi$.

The following two results determine interesting consistency properties, namely, that the method correctly identifies the region where the true kernel is null and vice versa. Specifically  Theorem~\ref{th:consistency1} assumes that the true kernel lies in the vector space generated by the tensor product $V_{\theta \otimes \varphi}$ and Theorem~~\ref{th:consistency2}  relaxes this assumption.  Both theorems hold under the correct model specification and when $e_i(s)$ are uncorrelated Gaussian errors with variance $\sigma^2>0$. 
Both results are adaptations of the consistency results of \citet{jenatton2011structured}.

Before stating the results, we define the following quantities based on the assumption that there exists a true kernel $\psi_T$. In Theorem \ref{th:consistency1} we assume that $\psi_T$ belongs   to the vector space $V_{\varphi \otimes \theta}$ and thus that there exist a unique  matrix of coefficients $\Psib_T$ representing it.  Furthermore, let ${\cal B }\subset \{1, \dots, B, B+1\}$ be the set of indices of the true non-zero groups and ${\cal I} \subset \{(m,l): m=1, \dots, M; \, l = 1, \dots, L\}$ the set of the indices of the non-zero coefficients $\psi_{m,l}$. Let ${\mathbf{Z}}_1$ the matrix containing the columns of ${\mathbf{Z}}$ associated to the true non zero coefficients and ${\mathbf{Z}}_0$ its complement. Define the norm
\[
\Omega^c_{\cal B}( u ) = \sum_{b \in \cal B}  || c_{b}^{({\cal B}^C)} \odot u ||_2 ,
\]
which is the modification of the penalising norm of \eqref{eq:min2} considering the sum over the inactive groups. Let further $\left(\Omega^c_{\cal B} \right)^*$ be its dual norm. Finally let $\mathbf{r}_1$ be the vector containing, for $(m,l) \in \cal I$, the elements 
\[
\psi_{m,l} c^2_{m,l} \sum_{b  \in \cal B} || c_b \odot \bpsi ||_2^{-1}.
\]
%
\begin{theorem}
	\label{th:consistency1}
Let $\psi_T$ be the true HS operator with  $\psi_T \in V_{\varphi \otimes \theta}$ where  $V_{\varphi \otimes \theta}$ is the vector space defined by the tensor product of the two B-splines basis $\theta$ and $\varphi$ both with fixed dimensions $L$ and $M$, respectively.   Let $\F_{0T}$ be the induced true sparsity set, i.e. $\psi_T(t,s)=0$ for each $(t,s) \in \F_{0T}$. 
Define the event $E_0$ as
\[
E_0 = \left\{ 
\int_{\F_{0T}} \left[\widehat{\psi}_\lambda (t,s)\right]^2 dt ds =0, \quad 	
\left(\inf_{\F_1 \subset \F_{0T}^C} 
\int_{\F_{1}} \left[\widehat{\psi}_\lambda (t,s)\right]^2 dt ds\right) >0
\right\}.
\]
For $\lambda >0$, $nG>LM$,  $\lambda \to 0$,  $\lambda \sqrt{n}\to \infty$,    
$
\left(\Omega^c_{\cal B} \right)^*\left( {\mathbf{Z}}_0^T {\mathbf{Z}}_1 ({\mathbf{Z}}_1^T {\mathbf{Z}}_1)^{-1} \mathbf{r}_1  \right) <1,
$
 then $\lim_{n\to \infty} \mathbb{P}(E_0) = 1.$
\label{th:consistencyfixeddim}
\end{theorem}

The next theorem avoids the strict assumption of $\psi_T \in V_{\varphi \otimes \theta}$. The main idea is to introduce suitable approximations of $\psi_T$ and the induced $ \F_{0T}$ and study  consistency of the estimator to those approximations. Specifically, we let $\bar \F_{0T}$ be the subset of $\F_{0T}$ of maximum Lebesgue measure  belonging to the power set ${\cal P}(F)$ and define  $\psi_T^P$ as
\[
\psi_T^P = \arg\min_{\psi \in V_{\varphi \otimes \theta}; \F_{0\psi} = \bar \F_{0T}} 
\int_\F \left(\psi(t,s) - \psi_T(t,s)\right)^2dtds.
\]
In what follows $\cal B, I,$ ${\mathbf{Z}}_1$, ${\mathbf{Z}}_2$ and the related norm $\Omega_{\cal B}^c$ and its dual, are defined with respect to the unique  $\Psib_T^P$ associated to $\psi_T^P$.
%
\begin{theorem}
	\label{th:consistency2}
	Let $\psi_T$ be the true HS operator and  $\F_{0T}$ the true induced sparsity set, i.e. $\psi_T(t,s)=0$ for each $(t,s) \in \F_{0T}$. 
	Define the event $E_k$ as
\[
E_k = \left\{ 
\int_{\F_{0T}} \left[\widehat{\psi}_\lambda (t,s)\right]^2 dt ds \leq \frac{k}{LM}, \quad 	
\left(\inf_{\F_1 \subset \F_{0T}^C} 
\int_{\F_{1}} \left[\widehat{\psi}_\lambda (t,s)\right]^2 dt ds\right) >0
\right\},
\]
with $k$ a positive constant.	Let $D<LM$ be the number of columns of ${\mathbf{Z}}_1$.  For  $\tau >0$, assume $\left(\Omega^c_{\cal B} \right)^*\left( {\mathbf{Z}}_0^T {\mathbf{Z}}_1 ({\mathbf{Z}}_1^T {\mathbf{Z}}_1)^{-1} \mathbf{r}_1  \right) <1-\tau,$ with
\begin{equation}
	\frac{\mu(\F_{0T}^C)}{\mu (\F)} \leq  \frac{(nG-1)(d+1)^2}{(L-1)(M-1)}.
	\label{eq:sparsitycondition}
\end{equation}
If $\tau \lambda \sqrt{n} \geq \sigma C_3$ and $\lambda\sqrt{D}\leq C_4$, then the probability of the event $E_k$ is lower bounded by
\[
1-\left(\exp\left\{-\frac{n\lambda^2\tau^2 C_1}{2\sigma^2}\right\}
+2 D\exp\left\{-\frac{n\lambda^2\tau^2 C_2}{2D\sigma^2}\right\}\right),
\] 
where $C_1, C_2, C_3,$ and $C_4$ are positive constants. 
\end{theorem}

%% file: 5_simstudy.tex
\section{Simulation study}
\label{sec:sim}

In this section we assess the empirical performance of our LSFR model by means of simulations. We generated a sample of a functional covariate $x_1(t), \ldots, x_n(t)$ from cubic B-splines basis with 15 evenly spaced knots between 0 and 1. These functional data are assumed to be observed on an equispaced grid of $100$ points and the sample size $n$ is $n\in \{50, 150\}$.
Conditionally on these functional covariates we generated functional responses as $y_i(s) = \int \psi_k(t,s)x_i(t) dt + e_i(s)$ under four different levels of sparsity of the operator $\psi_k$ and two different signal-to-noise ratios.

The first and the second scenario focuses on a quasi--concurrent and historical relations, respectively. 
The third scenario is associated to a local dependence between the dependent variable and the functional covariate with the region $\mathcal{F}_0$ being the union of two non--intersecting rectangles. Finally,   last scenario assumes no sparsity. Graphical representations of the kernels $\psi_k$ are showed in  Figure \ref{fig:psi3d}.
%
%
We vary the signal-to-noise ratio  $\mbox{SNR} \in \{2, 4\}$ with 
\[
\mbox{SNR} = \frac{1}{n} \sum_{i=1}^n \left\{ \int \left[\int \psi(t,s)x_i(t) dt\right]^2 ds \ \bigg/  \int e_i^2(s)ds \right\}.
\]
We choose $d=4$ and $L=M=20$ leading to $400$ coefficients in $\Psib$.
For each combination of $\psi_k$, $n$, and \mbox{\mbox{SNR}} we simulated $R=100$ replicated datasets, leading to $4 \times 2 \times 2 \times 100 = 1600$ independent datasets.

We compare our LSFR with the results obtained minimizing  \eqref{eq:min1} with Ridge, Lasso, and Elastic-net penalties \citep{hastierr}. Classical group-Lasso penalty has not been fitted due to its limitations discussed in Section \ref{sec:kerbasis}. We expect the Ridge approach to be the best competitor in the last scenario and in general if estimation is evaluated in the non-sparse regions of the true kernels. On the other side we expect Lasso to be the best competitor in detecting the sparse regions.  The Elastic-net minimization and our LSFR solution are expected to take the best of both Lasso and Ridge.

All methods depend on the choice of the tuning parameter $\lambda$. We run each model  on a fine grid of $\lambda$ and choose as final estimate for each method, the one minimizing the prediction error for the general estimator $\breve \psi_\lambda$ of $\psi$, i.e. 
\[
 \breve \psi = \arg \min_\lambda \sum_{i=1}^{200} \bigg(y_i(s) - \int \breve{\psi}_{\lambda}(t,s) x_i(t) dt \bigg)^2,
\]
where the sum is defined on external validation set of size 200. The range of possible $\lambda$ depends on the specific methods and for LFSR we refer to \eqref{eq:ovglasso_lambdamax} in Appendix. For all competing methods we rely on their implementation in the package \texttt{glmnet} \citep{glmnet}.

\begin{figure}[t]
    \centering
    \begin{minipage}{.24\textwidth}
    \centering
    \includegraphics[width=1\textwidth]{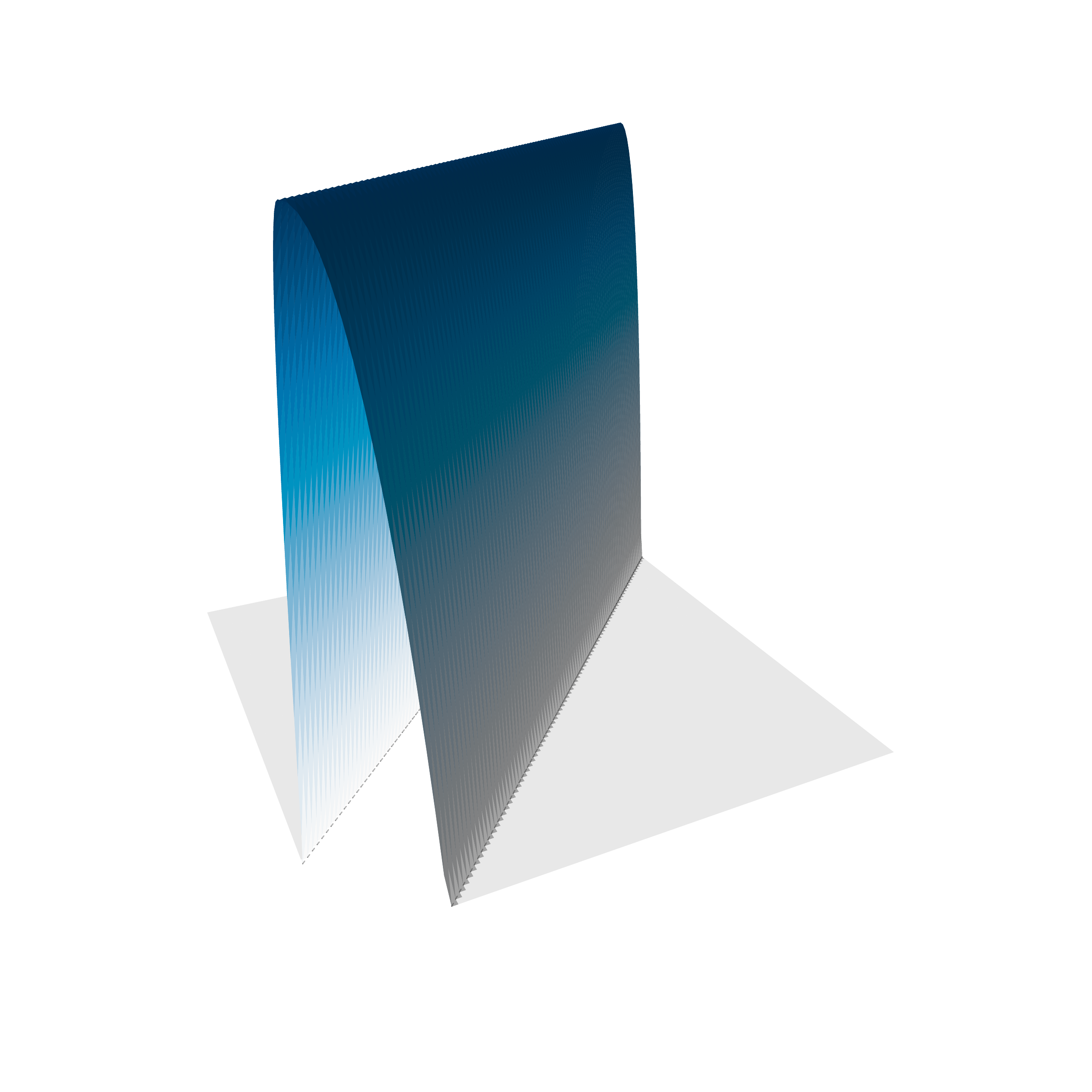}
    \end{minipage}%
    \begin{minipage}{.24\textwidth}
    \centering
    \includegraphics[width=1\textwidth]{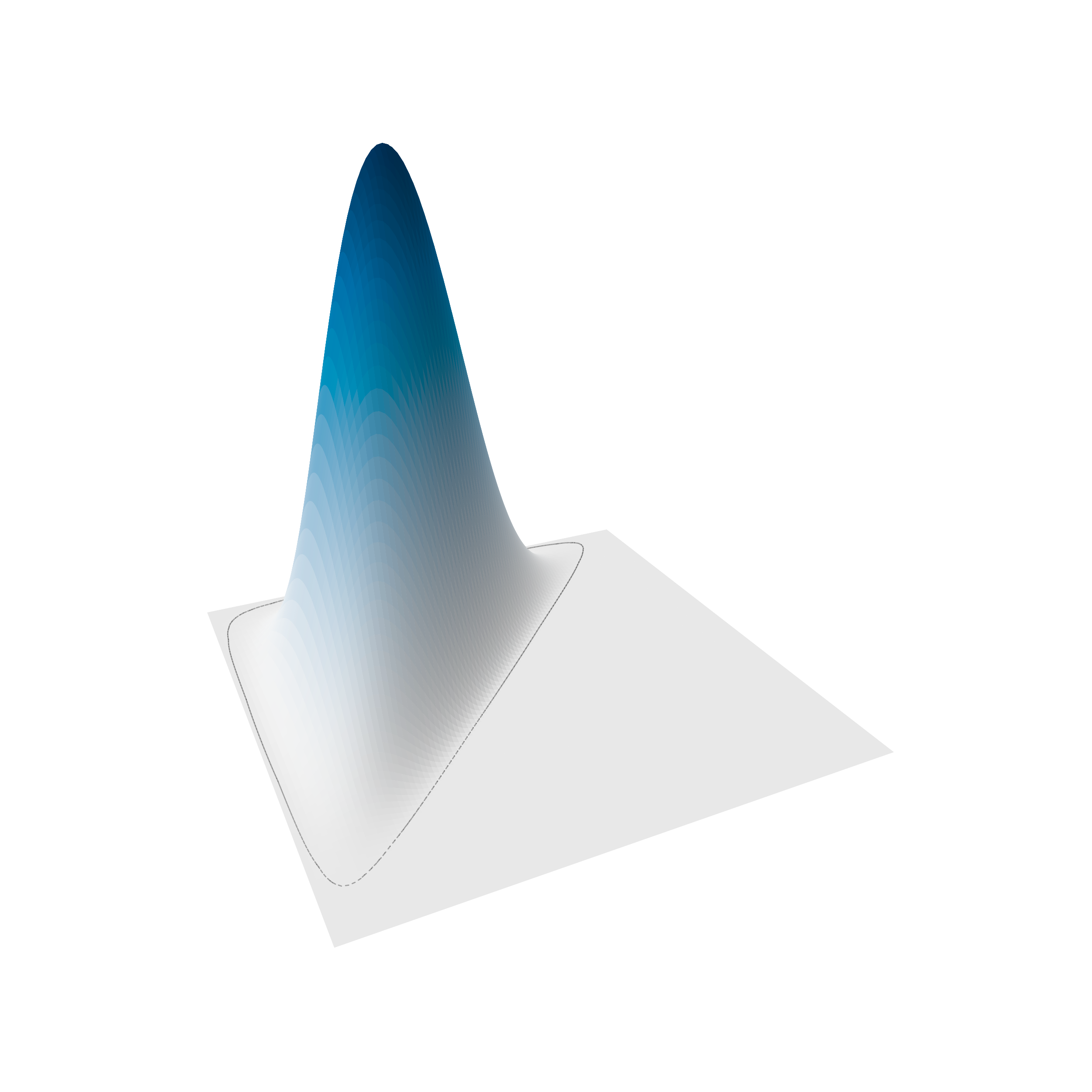}
    \end{minipage}%
    \begin{minipage}{.24\textwidth}
    \centering
    \vspace*{1cm}
    \includegraphics[width=1\textwidth]{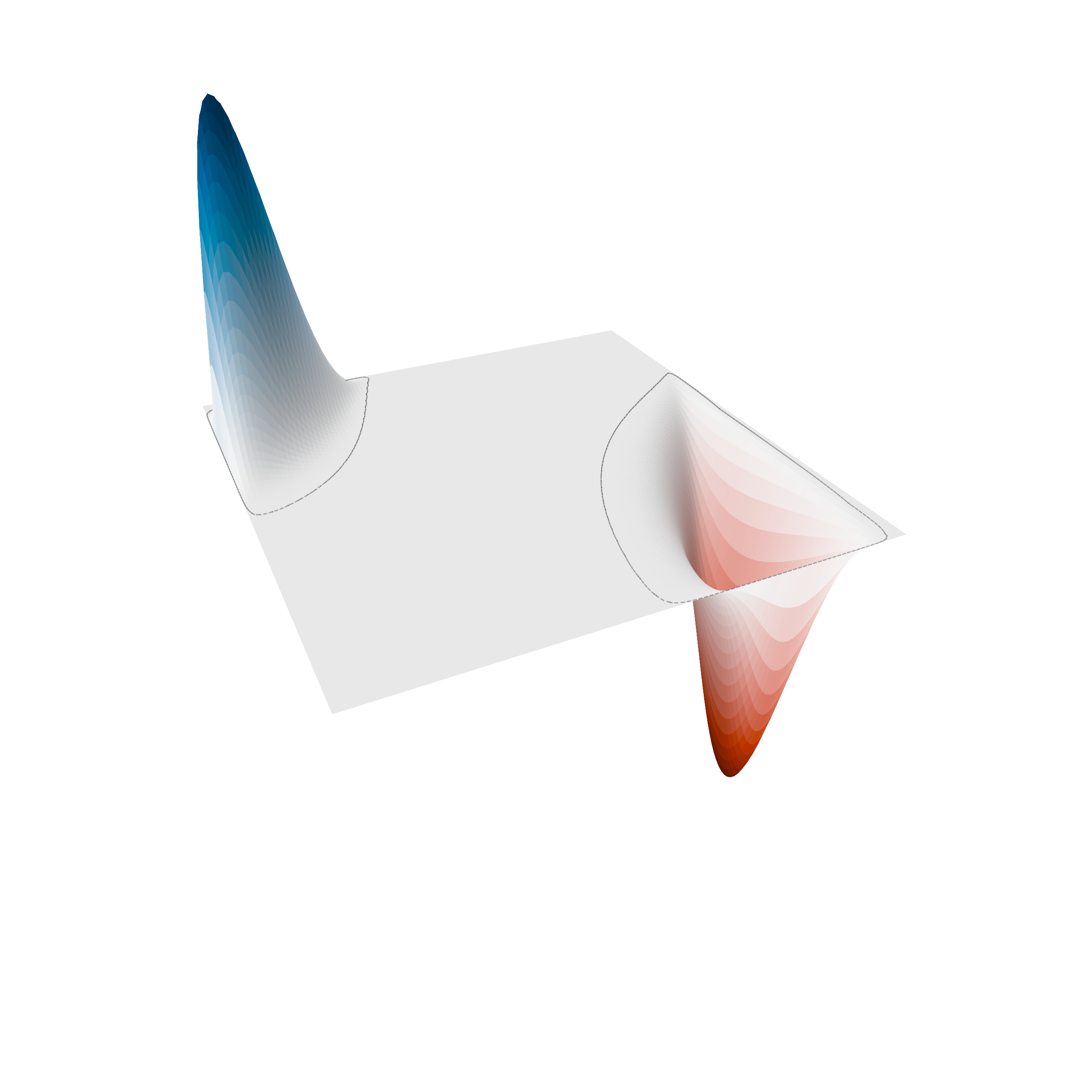}
    \end{minipage}
    \begin{minipage}{.24\textwidth}
    \centering
    \includegraphics[width=1\textwidth]{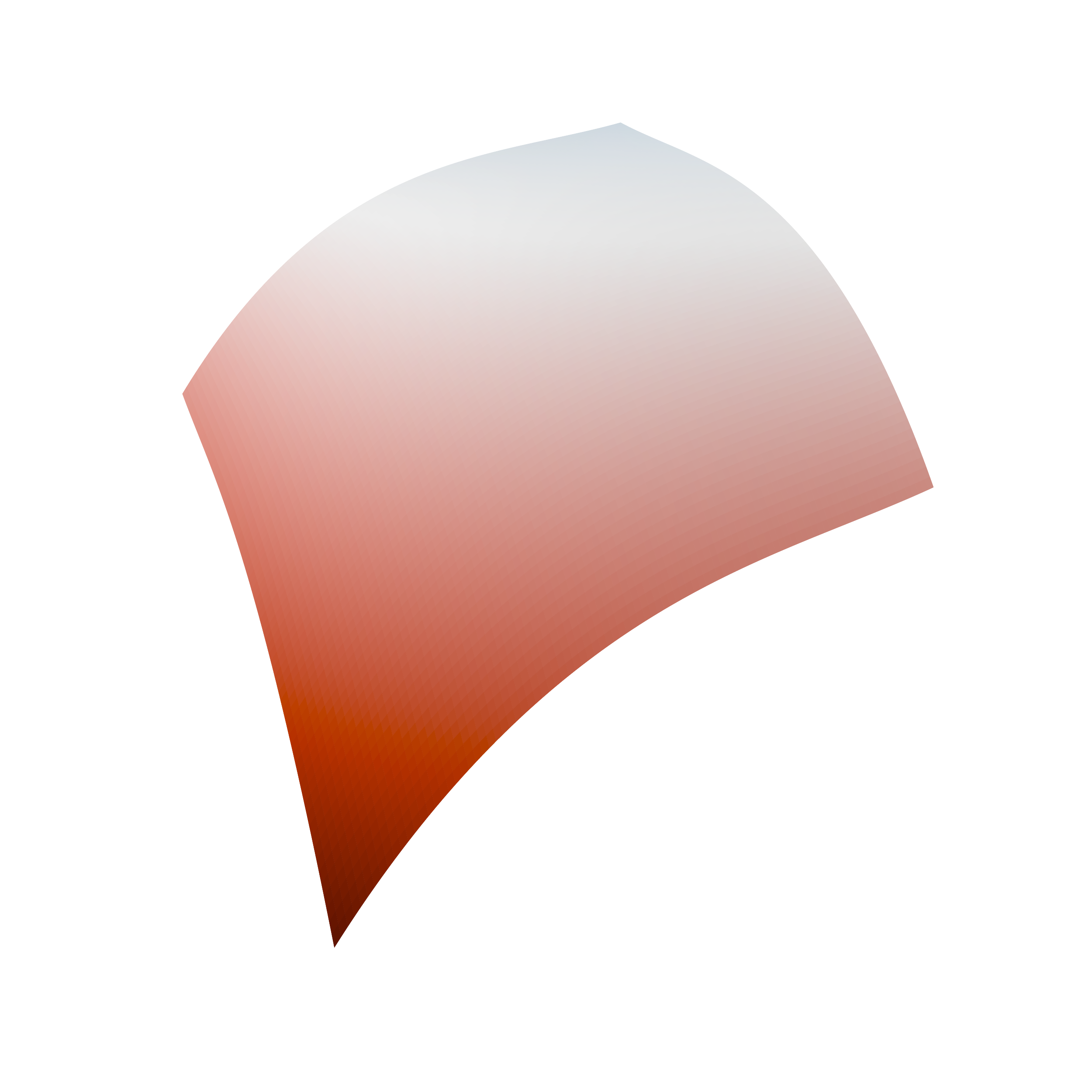}
    \end{minipage}%
    \caption{Three dimensional representations of $\psi_k$ with, from left to right, $k=1, 2, 3, 4$.}
    \label{fig:psi3d}
\end{figure}


The final estimates are  evaluated in both subsets of $\mathcal{F}$ where the true function is zero or not, on a third independent test set of size 1{,}000. Specifically, we define two partial integrated square error measures as
\begin{equation}
ISE_0(\breve \psi)  = \int_{\mathcal{F}_0} \bigg( \breve{\psi}(t, s) - \psi(t, s) \bigg)^2 dt ds, \quad 
ISE_1(\breve \psi)  = \int_{\mathcal{F}_0^C} \bigg( \breve{\psi}(t, s) - \psi(t, s) \bigg)^2 dt ds,
\end{equation}
representing the error in not estimating the sparsity of the true kernel  and the estimation error when the true kernel is smooth, respectively. As a global measure, we use the integrated squared error on the whole $\F$, which can be obtained as a suitable convex combination of the previous measures, i.e.
\begin{equation}
ISE(\breve{\psi})  = ISE_0(\breve \psi)  \mu(\mathcal{F}_0) + ISE_1(\breve \psi)  \mu(\mathcal{F}_1).
\end{equation}

\input{tabellone2}

\subsection{Results}

Table 1 reports $ISE_0$ and $ISE_1$ for four different settings. As default setting, we consider  $n=50$, $\mbox{SNR}=4$ and $M=L=20$. As second setting, we decrease the $\mbox{SNR}=2$ while keeping fixed $n=50$ and $M=L=20$. In the third, we assess the effect of changing the resolution of the kernels' basis representations letting $M=L=40$ with $n=50$ and $\mbox{SNR}=4$. Finally, we assess the effect of a bigger sample size letting $n=150$ while keeping the default $\mbox{SNR}=4$ and $M=L=20$. For all these settings, the performance of each method is evaluated for all the four kernels. 

The proposed LSFR model exhibits good performance consistently in all settings and considering both the evaluation metrics. Specifically, our approach achieves a uniformly lower $ISE_0$ both with respect to Ridge and Elastic-net in all settings while presenting similar performance with respect to  Lasso. In fact the results in terms of $ISE_0$ are broadly comparable for our LSFR and the Lasso for the estimation of the first and third kernel, while our LSFR  is even slightly better in estimating the second kernel. Looking at  $ISE_1$, as expected, the performance of the Lasso and of the Ridge and Elastic-net are flipped. The proposed LSFR model, however, consistently outperforms all competing methods also in terms of $ISE_1$.

Different values of \mbox{\mbox{SNR}}, $n$, or $M$ and $L$ have different impacts on the results. Lowering the signal-to-noise ratio or increasing the sample size reduces ---resp. increases--- the errors almost proportionally among different estimation procedures. On the contrary, changing the number of basis from 20 to 40 has a relevant impact on the results. Particularly, a dramatic increase of  $ISE_1$ is observed for  Lasso, in all scenarios. This can be explained by the fact that using a larger number of bases is equivalent to reduce the support of each element of the basis. Consistently with this,  Lasso penalizes coefficients that are much more correlated, yielding to an erratic sequence of estimated coefficients. Conversely, the larger number of bases of the third setting does not negatively affect the remaining methods including our LSFR model, but instead it allows for a potentially better identification of small sparse regions impossible to detect with a lower number of bases. 

Figure \ref{fig2} reports the boxplots of the global relative efficiency of the three competing methods, defined as the ratio between the three different $ISE$ and the $ISE$ achieved by our LSFR model in the default setting. All boxplots are above one,  witnessing that the proposed method uniformly attains a lower $ISE$. Exception made for the third scenario, Lasso is the worst method, while Elastic-net is the most flexible, as expected. Similar results can be noticed for the other simulation settings. 

\begin{figure}[t]
\centering
\includegraphics[scale=.4]{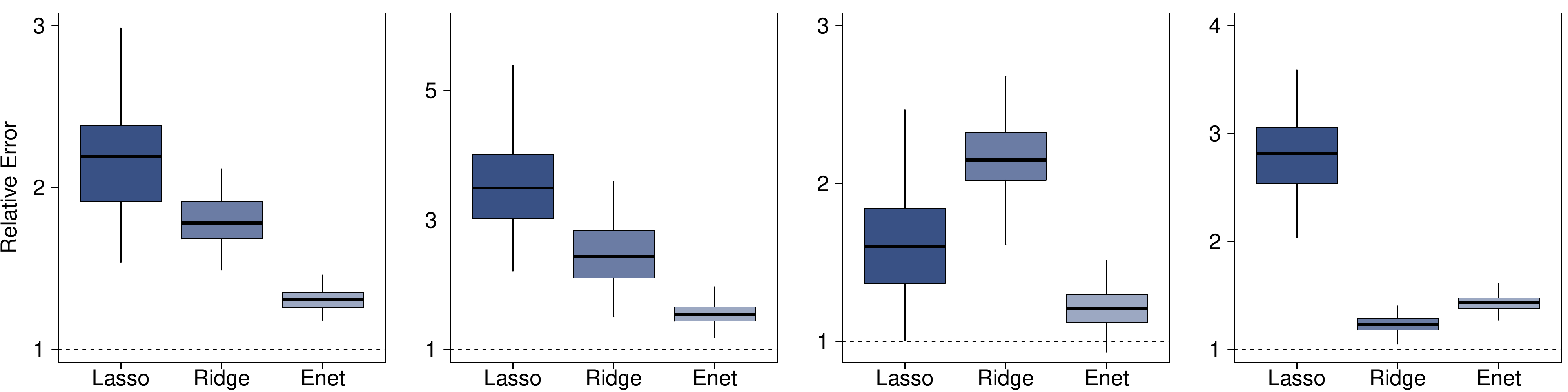}
\caption{Relative ISE for the simulation study under the default setting and for the four scenarios.}
\label{fig2}
\end{figure}

%% file: tabellone2.tex
\singlespacing
\begin{table}
\caption{Monte Carlo means and standard deviations ($\times 10^5$) of the $ISE_0$ and $ISE_1$ between the estimated coefficients  and the true coefficients $\psi_k$, $k=1,2,3,4$ over $100$ replicates.
Rows labelled S1 report the results for the default simulation settings, i.e. $n=50$, $SNR=4$, $L=M=20$; Rows labelled S2 report the results for $SNR =2$ and default remaining settings; 	Rows labelled S3 report the results for $L=M=40$ and default remaining settings; Rows labelled S4 report the results for $n = 150$ and default remaining settings.}
	\centering
	\resizebox{\textwidth}{!}{
	\begin{tabular}{crrrrrrrr}\toprule
		& LSFR & & Lasso & & Ridge & & Enet & \\
		\cline{2-9}\\
		& \multicolumn{1}{c}{$ISE_0$} & 
		\multicolumn{1}{c}{\textit{$ISE_1$}} & 
		\multicolumn{1}{c}{$ISE_0$} & 
		\multicolumn{1}{c}{\textit{$ISE_1$}} & 
		\multicolumn{1}{c}{$ISE_0$} & 
		\multicolumn{1}{c}{\textit{$ISE_1$}} & 
		\multicolumn{1}{c}{$ISE_0$} &
		\multicolumn{1}{c}{\textit{$ISE_1$}} \\
		\cline{2-3}
		\cline{4-5}
		\cline{6-7}
		\cline{8-9}\\
		{\small $\psi_1$}\\
		\vspace*{.1cm}
		S1 & 0.18 \scriptsize{(0.04)}& 1.07 \scriptsize{(0.20)}&  0.18 \scriptsize{(0.06)}&  2.76 \scriptsize{(0.44)}& 0.66 \scriptsize{(0.11)} & 1.12 \scriptsize{(0.22)}& 0.26 \scriptsize{(0.06)} & 1.34 \scriptsize{(0.23)}  \\
		\vspace*{.1cm}
		S2   & 0.28 \scriptsize{(0.08)}&  1.63 \scriptsize{(0.34)}& 0.31 \scriptsize{(0.12)}& 4.08 \scriptsize{(0.67)}& 0.98 \scriptsize{(0.17)}& 1.74 \scriptsize{(0.35)}& 0.42 \scriptsize{(0.12)} & 2.03 \scriptsize{(0.40)}   \\
		\vspace*{.1cm}
		 S3   & 0.14 \scriptsize{(0.05)}& 1.66 \scriptsize{(0.36)} & 0.37  \scriptsize{(0.13)}& 16.24 \scriptsize{(2.06)}& 1.02 \scriptsize{(0.20)}& 1.48 \scriptsize{(0.31)}& 0.39 \scriptsize{(0.11)} & 2.08 \scriptsize{(0.40)}  \\
		\vspace*{.1cm}
		S4   & 0.11 \scriptsize{(0.02)}& 0.58 \scriptsize{(0.09)} & 0.08  \scriptsize{(0.03)}& 1.45 \scriptsize{(0.24)}& 0.34 \scriptsize{(0.06)}& 0.56 \scriptsize{(0.09)}& 0.14 \scriptsize{(0.03)} & 0.71 \scriptsize{(0.11)}  \\
		\vspace*{.1cm}
		{\small $\psi_2$}\\
		\vspace*{.1cm}
		S1 & 0.04 \scriptsize{(0.03)}& 0.51 \scriptsize{(0.14)}&  0.12 \scriptsize{(0.07)}& 1.80  \scriptsize{(0.37)}& 0.52 \scriptsize{(0.14)} & 0.75 \scriptsize{(0.18)}& 0.16 \scriptsize{(0.08)} & 0.68 \scriptsize{(0.17)}  \\
		\vspace*{.1cm}
		S2   & 0.08 \scriptsize{(0.05)}&  0.82 \scriptsize{(0.22)}& 0.23  \scriptsize{(0.13)}& 2.42 \scriptsize{(0.47)}& 0.86 \scriptsize{(0.21)}& 1.21 \scriptsize{(0.26)}& 0.30 \scriptsize{(0.13)} & 1.09 \scriptsize{(0.25)}   \\
		\vspace*{.1cm}
		 S3   & 0.08 \scriptsize{(0.05)}& 0.74 \scriptsize{(0.21)} & 0.32  \scriptsize{(0.16)}& 9.98 \scriptsize{(2.75)}& 0.96 \scriptsize{(0.25)}& 1.16 \scriptsize{(0.26)}& 0.39 \scriptsize{(0.14)} & 1.04 \scriptsize{(0.23)}  \\
		\vspace*{.1cm}
		S4   & 0.02 \scriptsize{(0.01)}& 0.23 \scriptsize{(0.05)} & 0.04  \scriptsize{(0.02)}& 1.12 \scriptsize{(0.22)}& 0.26 \scriptsize{(0.07)}& 0.33 \scriptsize{(0.07)}& 0.06 \scriptsize{(0.03)} & 0.29 \scriptsize{(0.06)}  \\
		\vspace*{.1cm}
		{\small $\psi_3$}\\
		\vspace*{.1cm}
		S1 &  0.26 \scriptsize{(0.06)}& 0.58 \scriptsize{(0.11)} &  0.25 \scriptsize {(0.09)}& 1.01 \scriptsize{(0.19)} & 0.86 \scriptsize{(0.19)} & 1.10 \scriptsize{(0.16)} & 0.32 \scriptsize{(0.09)} & 0.69 \scriptsize{(0.11)}  \\
		\vspace*{.1cm}
		S2 & 0.32 \scriptsize{(0.09)}&  0.82 \scriptsize{(0.20)} & 0.32  \scriptsize{(0.15)}& 1.29 \scriptsize{(0.27)}& 1.14 \scriptsize{(0.30)}& 1.57 \scriptsize{(0.27)} & 0.40 \scriptsize{(0.15)} & 0.96 \scriptsize{(0.20)}   \\
		\vspace*{.1cm}
		 S3  & 0.25 \scriptsize{(0.07)}&  0.62 \scriptsize{(0.14)} & 0.47 \scriptsize{(0.14)}&  7.17 \scriptsize{(1.45)}& 1.18 \scriptsize{(0.29)}& 1.40 \scriptsize{(0.20)}& 0.41 \scriptsize{(0.12)} & 0.91 \scriptsize{(0.17)}  \\
		\vspace*{.1cm}
		S4   & 0.22 \scriptsize{(0.04)}& 0.36 \scriptsize{(0.06)} & 0.21  \scriptsize{(0.05)}&  0.68 \scriptsize{(0.13)}& 0.56 \scriptsize{(0.11)}& 0.61 \scriptsize{(0.09)}& 0.27 \scriptsize{(0.07)} & 0.46 \scriptsize{(0.08)}  \\
		\vspace*{.1cm}
		{\small $\psi_4$}\\
		\vspace*{.1cm}
		 S1 & -- & 0.76 \scriptsize{(0.11)}&  -- & 2.13 \scriptsize{(0.26)}& --  & 0.93 \scriptsize{(0.12)}& --  & 1.08 \scriptsize{(0.13)}  \\
		\vspace*{.1cm}
		 S2  & -- & 1.07 \scriptsize{(0.16)}& --  & 2.62 \scriptsize{(0.26)}& -- & 1.30   \scriptsize{(0.17)}& --  & 1.48 \scriptsize{(0.17)}   \\
		\vspace*{.1cm}
		 S3  & -- & 1.04 \scriptsize{(0.18)} & -- & 5.87 \scriptsize{(0.81)}& -- & 1.10 \scriptsize{(0.18)}& --  & 1.34 \scriptsize{(0.20)}  \\
		\vspace*{.1cm}
		 S4   & -- & 0.42 \scriptsize{(0.06)} & --  & 1.09 \scriptsize{(0.17)}& --& 0.53 \scriptsize{(0.06)}& --  & 0.61 \scriptsize{(0.08)}  \\
\bottomrule
	\end{tabular}
	}
	\label{tab:tabellauno}
\end{table} 
\doublespacing

%% file: 6_application.tex
\section{Applications}
\label{sec:app}

\subsection{Swedish Mortality}
\label{sec:swedish}

We apply the proposed LSFR model to the well-known Swedish Mortality dataset, available from the Human Mortality Database and considered one of the most reliable dataset on long-term longitudinal mortality. We focus on the analysis of the log-hazard rate functions of the Swedish female population between the years 1751 and 1894. The goal of our analysis is to model the log-hazard function $y_i$ at a specific calendar year $i$ by using the log-hazard function at previous year $y_{i-1}$. The log-hazard rate for year $i$ and age $s$ is computed as the logarithm of the ratio between women born on year $i$ who died at age $s$ and women born on year $i$ still alive at age $s$. 
An underlying autoregressive linear relation between $y_{i-1}$ and $y_i$ is assumed and specifically
\[
y_i(s) = \int y_{i-1}(t)\psi(t,s) dt + e_{i}(s).
\]
The estimated kernel can be interpreted as the influence of the log-hazard rate at year $i-1$ and age $t$ on the log-hazard rate at year $i$ and age $s$. Existing studies \citep{FDAR, chiou_muller} show that the hazard function at year $i$ and age $s$ is mainly influenced by the hazard function at the  previous year $i-1$ at age $t = s-1$, resembling a quasi--concurrent relation. However, none of these studies reports the total absence of relation when $t$ and $s$ are far away and the corresponding estimated surfaces exhibit nonvanishing fluctuations even near the boundaries of their bivariate domains. 

\begin{figure}[t]
    \centering
    \begin{minipage}{.4\textwidth}
    \vspace*{-1.7cm}
    \hspace*{-1.5cm}
    \includegraphics[width=1.3\textwidth]{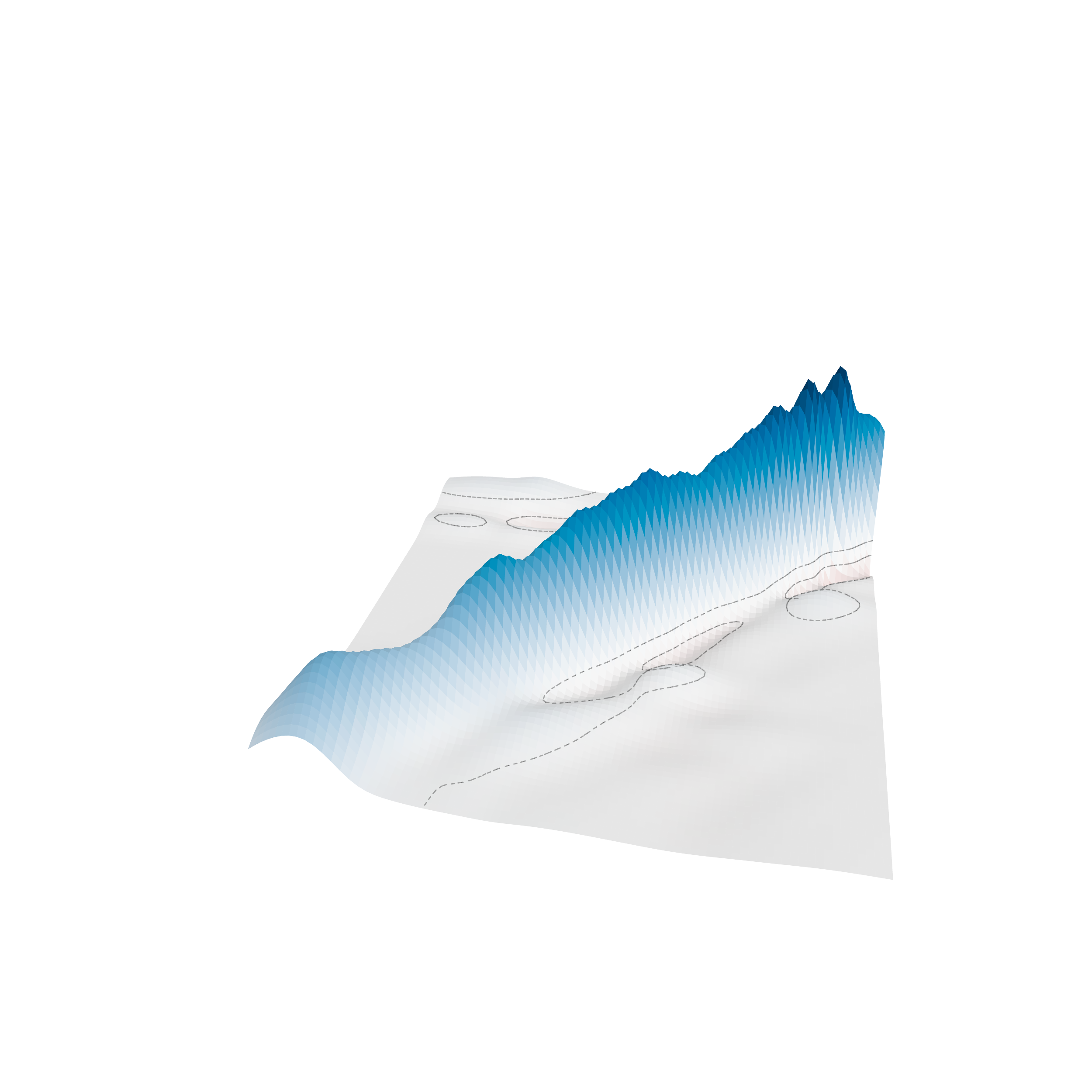}
    \end{minipage}%
    \begin{minipage}{.4\textwidth}
    \centering
    \includegraphics[width=.8\textwidth]{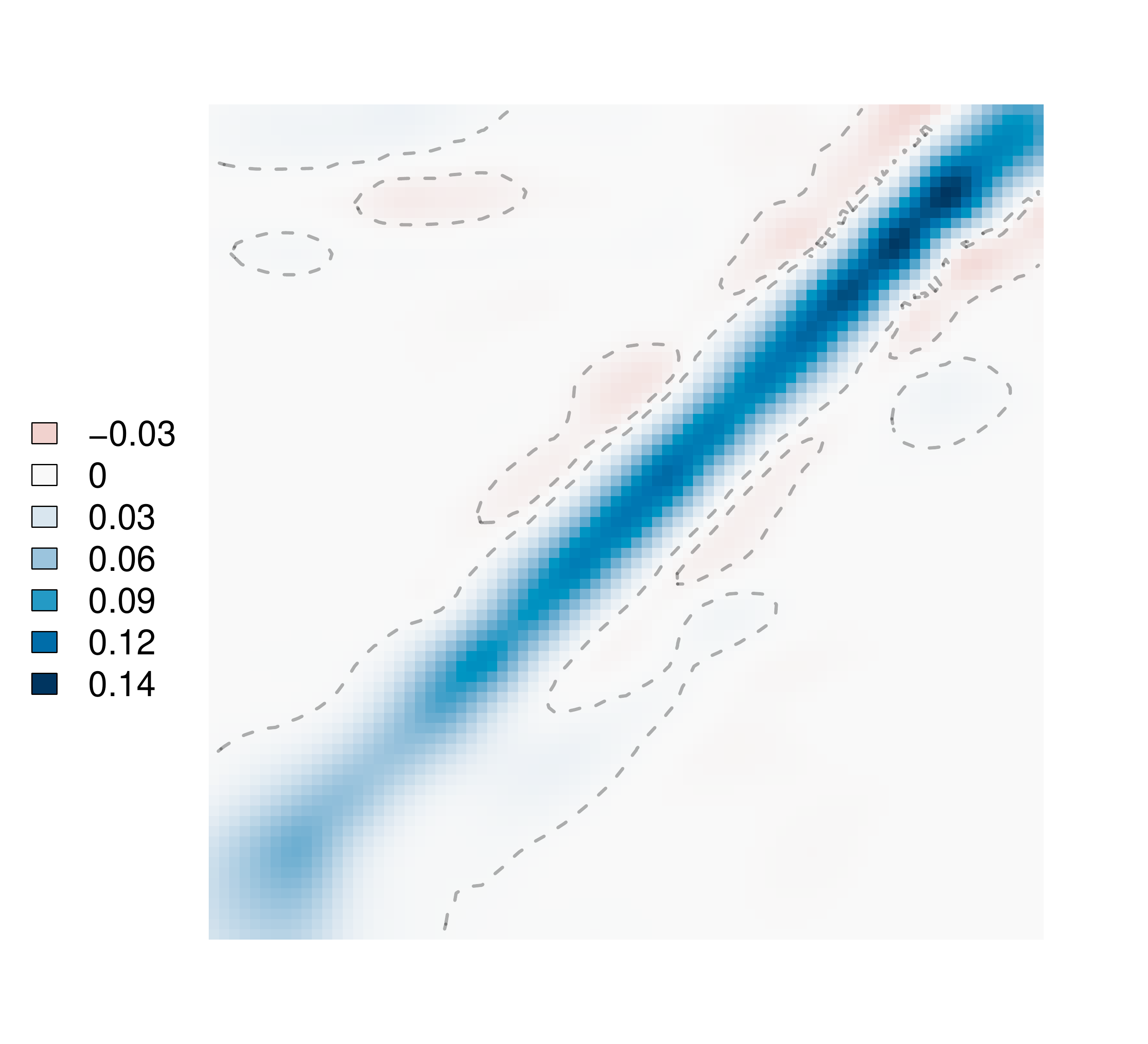}
    \end{minipage}
    \caption{Estimated regression surface for Swedish Mortality data. Left: perspective plot. Right: heatmap.}
    \label{fig:swedish}
\end{figure}

We implement the proposed locally sparse estimator by minimizing the objective (\ref{eq:min3_init}) with $d=4$, $M = L = 20$ basis functions on each dimension and we select the optimal value of $\lambda$ by means of cross--validation. Perspective and contour plots of the estimated kernel are depicted in Figure~\ref{fig:swedish}. 

Our estimate shows a marked positive diagonal confirming the positive influence on the log-hazard rate at age $s$ of the previous year's curve evaluated on a neighborhood of $s$. At the same time, the flat zero regions outside the diagonal suggest that there is no influence of the curves evaluated at distant ages.
Our estimate is more regular than previous approaches and its qualitative interpretation sharper and easier. Refer,  for example, to Figure 10.11 of \cite{FDAR}.

This result witnesses the practical relevance of adopting the proposed approach. Indeed, the resulting estimates, while being reminiscent of a concurrent model---inheriting its ease of interpretation---it gives further insights and improves the fit, representing the desired intermediate  solution between the concurrent and nonconcurrent models.

\subsection{Italian Gas Market}

\begin{figure}[t]
    \centering
    \begin{minipage}{.45\textwidth}
    \centering
    \includegraphics[width=.8\textwidth]{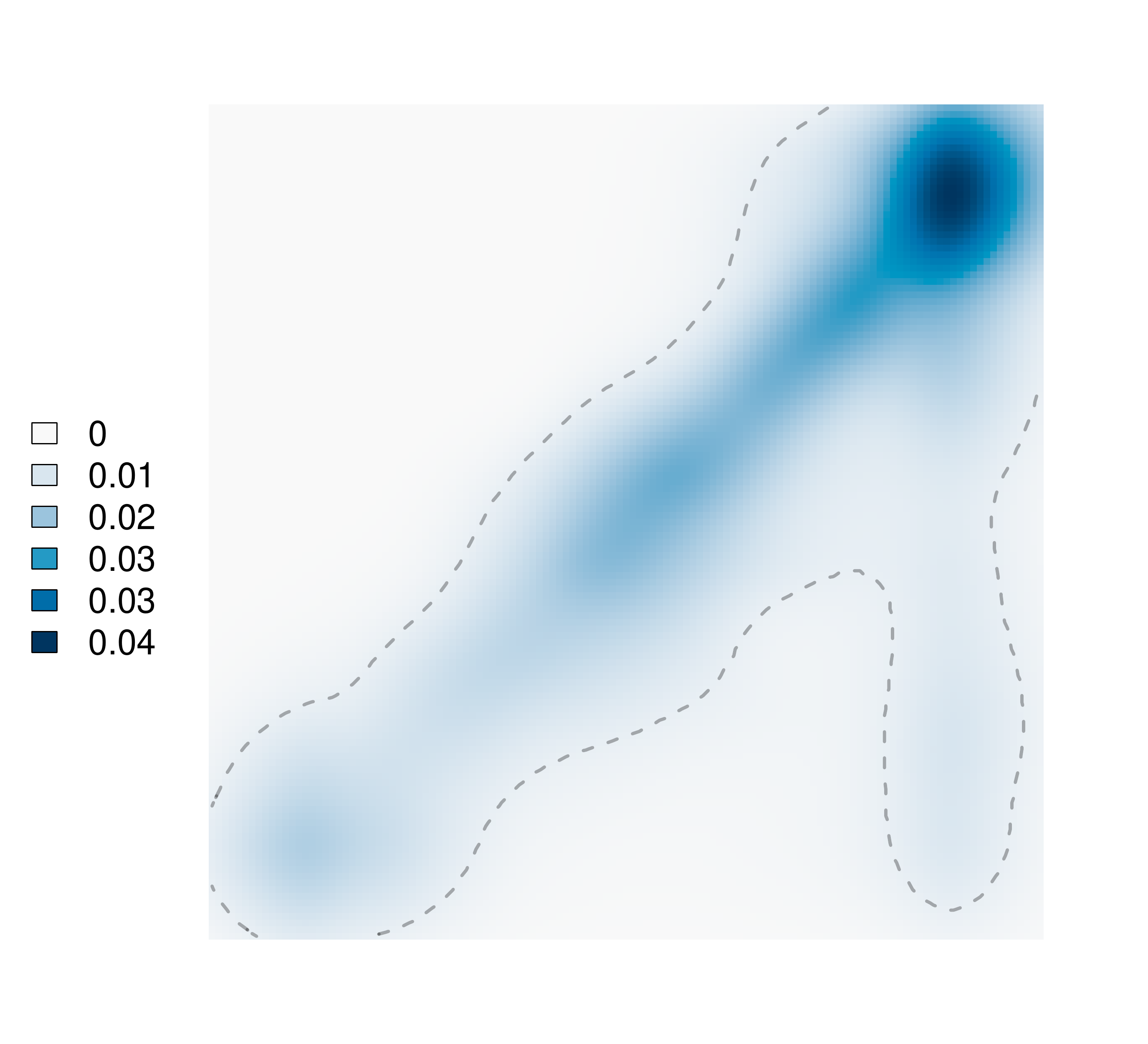}
    \end{minipage}
    \begin{minipage}{.45\textwidth}
    \centering
    \includegraphics[width=.8\textwidth]{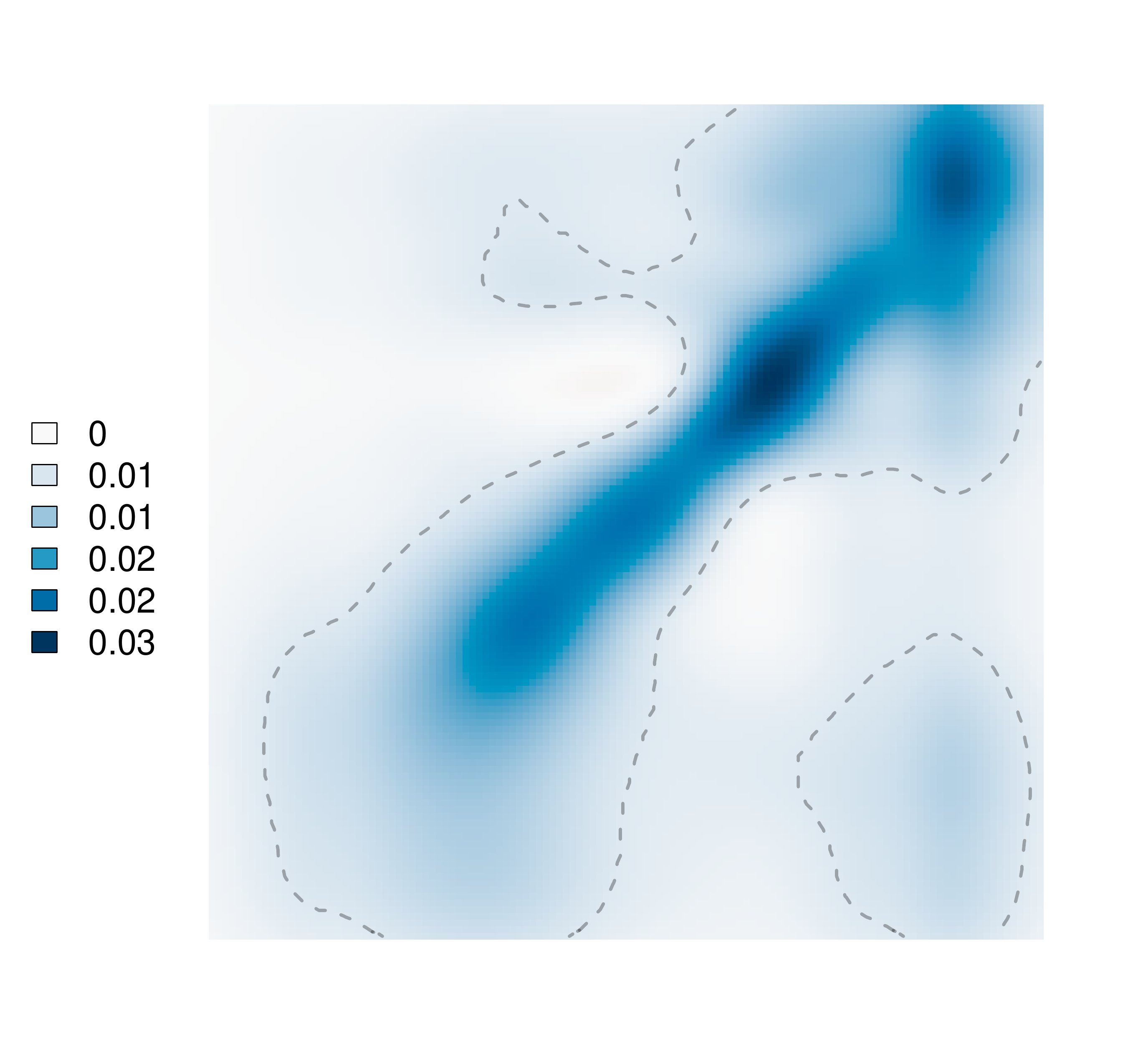}
    \end{minipage}
    \caption{Estimated regression surfaces for demand (left) and supply (right) curves in the Italian natural gas dataset. }
    \label{fig:gasmarket}
\end{figure}


As second illustrative example, we focus on the Italian natural gas data first analyzed by \citet{canale} by means of a concurrent functional time series regression fitted under a Ridge-like penalty. Data consists on 375 pairs of demand and supply curves related to the Italian gas trading platform in 2012. Details on how the curves are constructed from the original bids can be found in \citet{canale}. The domain of the demand and supply curves represent the quantity and the price of the natural gas, respectively. The demand (resp. supply) curves are subject to equality and inequality constraints at the two edges of the domain and are monotonically increasing (resp. decreasing).  Consistently with the so-called transform/back-transform method \citep{FDA1} we map the actual monotone and bounded curves using the $\mbox{logH}$ transformation introduced by \citet{canale} and fit a LSFR autoregressive model to the transformed data similarly to the previous section. Specifically, let $z_i$ be the demand (resp. supply) curve at day $i$. We let
\[
y_i(s) = 1 - \exp\left\{-\int_0^s\exp\{z_i(u)\} d u \right\}, 
\]
and then model transformed demand (resp. supply) curve $y_i$ at day $i$ with a first-order autoregressive linear model as done in Section \ref{sec:swedish}. We set $d=4$ and $M = L = 20$ basis functions and select $\lambda$ via  cross-validation, in the previous application. Two models for the demand and supply series are fitted separately. 

Figure \ref{fig:gasmarket} reports the contour plots of the estimated kernels. Similarly to Section \ref{sec:swedish}, both estimates present  marked positive diagonals while the influence of the curve at time $i-1$ on the curve at time $i$  vanishes if we move far away from the diagonal. The qualitative interpretation is clear: the price of a given quantity for  day $i$ mainly depends on the price of the $i-1$ demand curve in a neighborhood of the same quantity. 

Also in this practical case we stress the usefulness of the  proposed approach and its ability to lie between the concurrent and nonconcurrent approaches. 

%% file: 7_discussion.tex
\section{Discussion}
\label{sec:discussion}

	Motivated by the current gap between the concurrent and nonconcurrent models, we introduced a model for function-on-function regression which allows for local sparsity patterns, exploiting some of the properties of B-spline basis and a specifically tailored overlap group-Lasso penalization. 
	The proposed MM algorithm directly tackles the minimization arising from such a group-Lasso penalty, but it is also amenable to generalizations beyond the specific group structure motivated by our model and, furthermore, beyond the FDA context. 
	The empirical assessment through simulation and applications to real datasets, provides evidence of improved estimation accuracy with respect to standard competitors and straightforward and effortless interpretability of the results. The latter is of substantial interest when dealing with infinite dimensional objects as in FDA. 
	
	Despite we focused on the case of a single functional covariate, the extension of our modelling framework to $p$ covariates is straightforward. At the same time, our approach can be specified also in the simpler case when the $y_i$'s are scalar responses.

%% file: 8_appendix.tex
\newpage
\appendix
\section*{Appendix}

\subsection*{Examples of group-Lasso not allowed sparsity patterns}

\begin{figure}[H]
    \centering
    \includegraphics[width=.77\textwidth]{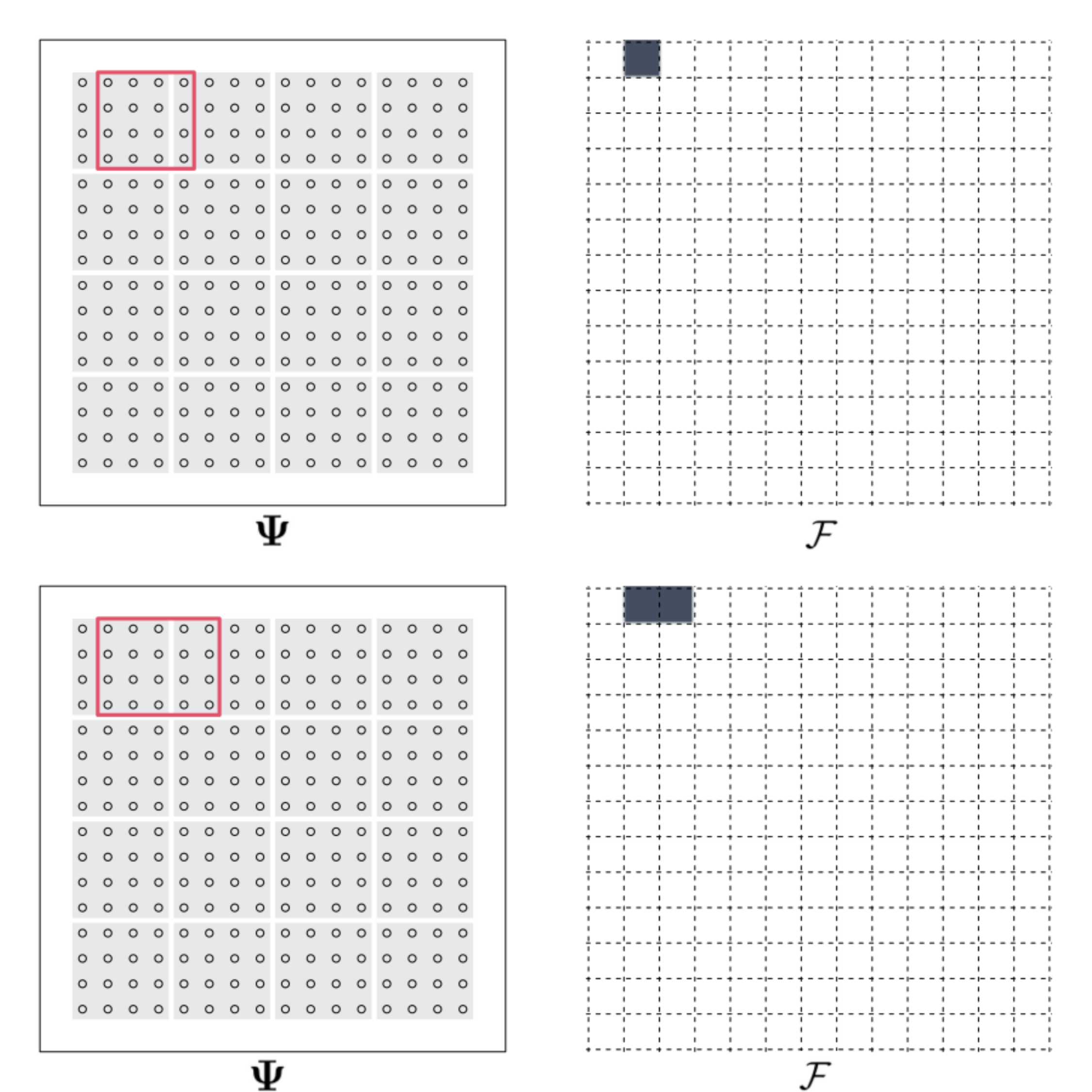}
    \caption{
    On the left, the matrices of coefficients $\Psib$  partitioned  in $d \times d$ ($d=4$) blocks (grey background) coherent with a group-Lasso approach. 
    On the right, the set $\mathcal{F}$ partitioned by the B-splines knots (dashed lines). Each square corresponds to a set $\mathcal{F}_{m,l}$. 
    In blue we depict two sparsity patterns ${\cal F}_0$ that are not allowed under the group-Lasso definition, consistent with Section 2.2. Specifically in the upper plot $\psi(t,s)=0$ for $(t,s) \in \mathcal{F}_{1,2}$ and no other neighboring sets. In the lower plot $\psi(t,s)=0$ for $(t,s) \in \mathcal{F}_{1,2} \cup \mathcal{F}_{1,3}$. To get these sparsity patterns, all the coefficient inside the red contours in the matrix  $\Psib$ (left plot) need to be set to zero. The red tiles are not consistent with the disjoint grouping (in grey). Notably, the overlapping group structure that we adopt allows for these situations. 
}
    \label{fig:esempi}
\end{figure}

\subsection*{Additional computational results}
In this Section we provide additional results concerning the MM algorithm introduced in Section \ref{sec:comp}. We first introduce the MM algorithm \ref{alg:mmalgo}, then we provide the exact computational costs and we conclude the section by considering the problem of finding the minimum value of $\lambda$ for which all the coefficients are equal to zero.
%
{\begin{center}
\begin{minipage}{.9\linewidth}
\begin{algorithm*}[H]
\label{alg:mmalgo}
\begingroup
    \fontsize{10pt}{11pt}\selectfont
\caption{MM algorithm for penalized functional regression} 
 \label{alg:coord_desc_mm_qr}
 {Set $k=0$, initialise the parameters $\widehat{\boldsymbol{\psi}}^{(0)}$ and evaluate $\ell\Big(\widehat{\boldsymbol{\psi}}^{(0)} \Big)$ as in equation \eqref{eq:min3_init}, compute $\mathbf{Z}^T\mathbf{y}$, $\mathbf{y}^T\mathbf{y}$ and $\mathbf{D}_b^2=\mathbf{D}_b^T\mathbf{D}_b$ for $b=1,2,\dots,B+1$, and set $\ell\Big(\widehat{\boldsymbol{\psi}}^{(-1)}\Big)=0$}\;   
 \While(){$\Big\vert\ell\Big(\widehat{\boldsymbol{\psi}}^{(k)} \Big)-\ell\Big(\widehat{\boldsymbol{\psi}}^{(k-1)} \Big)\Big\vert>\epsilon$}
 { 
{update $\widehat{\boldsymbol{\psi}}^{(k+1)}$}, as in equation \eqref{eq:minimization_surrogate_fun_sol}\;
 {update  $\widehat{d}_b^{(k)}$ as in equation \eqref{eq:mm_algo_update_constants} and compute $\mathbf{H}^{(k)}=\Big(\sum_{b=1}^{B+1} \widehat{d}_b^{(k)}\mathbf{D}_b^2\Big)^{1/2}$}\;
{evaluate $\ell\Big(\widehat{\boldsymbol{\psi}}^{(k+1)} \Big)$ as in equation \eqref{eq:min3_init}}\;
{set $k=k+1$}\;
}
\endgroup
\end{algorithm*}
\end{minipage}
\end{center}
}
%
\begin{proposition}
\label{prop:cost}
The exact computational cost in terms of floating point operations \citep[see, e.g.][]{golub_van_loan.2013} of performing Algorithm \ref{alg:coord_desc_mm_qr} is $13(nG)^3+LM\big[2(nG)^2+3nG+7B+15\big]+(LM)^2+7nG+6$ which is of order $\mathcal{O}\big((nG)^3\big)$.
\end{proposition}
\begin{proof}
The exact computational cost is
\begin{itemize}
    \item[-] $(LM+1)(2nG-1)$ to compute $\mathbf{Z}^T\mathbf{y}$ and $2LM+nG(2LM-1)+2nG+3LM(B+1)+2$ to compute $\ell(\boldsymbol{\psi}^{(-1)})$ and $2LM-1$ to compute $\mathbf{D}_b^2$,  in the initialization step, totalling $4nGLM+6LM+3BLM+2nG$;
    \item[-] $3LM+2$ to compute $\widehat{d}_n^{(k)}$ and $LM(B+1)+1$ to compute $\mathbf{H}^{(k)}$, totalling $(B+4)LM+3$;
    \item[-] $13(nG)^3$ to compute the spectral decomposition of the matrix $\mathbf{J}^{(k)}$ leveraging the Francis' algorithm \citep[][]{watkins.2011,golub_uhlig.2009},  \citep[see][]{trefethen_bau.1997};
    \item[-] $4nG+(LM)^2$ to compute $(\mathbf{B}^{(k)}_{\lambda})^{1/2}$;
    \item[-] $(2nG-1)nGLM+(LM)^2+3LM+2$ to compute $\widehat{\boldsymbol{\psi}}^{(k+1)}$, totalling $(2(nG)^2+nG+3)LM+(LM)^2+3$;
    \item[-] $2LM+nG(2LM-1)+2nG+3LM(B+1)+2$ to compute $\ell(\boldsymbol{\psi}^{(k+1)})$, totalling $(2+2nG)LM+nG+3LM(B+1)$,
\end{itemize}
totalling $13(nG)^3+LM\big[2(nG)^2+3nG+7B+15\big]+(LM)^2+7nG+6$.
\end{proof}
It is worth noting that the leading element of the computational cost is cubic in $nG$ which is assumed to be less than $LM$. Moreover, leveraging the Shermann-Morrison-Woodbury formula has the further benefit that, as the MM algorithm iterates, the columns of the design matrix $\mathbf{Z}$ that correspond to groups of coefficients that have been previously set to zero are discarded.\par
The next proposition provides the minimum value of the tuning parameter $\lambda$ for which all the coefficients in the minimization problem in Equation \eqref{eq:min3_init} are equal to zero.
\begin{proposition} 
\label{prop_ovgLasso_lambda_min}
\label{prop:smallestlambda}
For the overlap group-Lasso problem in equation \eqref{eq:min3_init} the smallest $\lambda$ at which all coefficients $\widehat{\boldsymbol{\psi}}_{\lambda_{\mathrm{max}}}$ are equal to zero is
\begin{equation}
\label{eq:ovglasso_lambdamax}
\widehat{\lambda}_{\max}=\max_{j}\Bigg\{\frac{\vert \mathbf{Z}_{j}^T\mathbf{y}\vert}{\widetilde{\boldsymbol{D}}_{jj}}\Bigg\},
\end{equation}
where $\mathbf{Z}_{j}$ denotes the $j$-th column vector of the matrix $\mathbf{Z}$ and $\widetilde{\boldsymbol{D}}_{jj}$ denotes the $j$-th diagonal element of the matrix $\widetilde{\mathbf{D}}=\sum_{b=1}^{B+1} \mathbf{D}^2_{b}$.
\end{proposition}
\begin{proof}
The objective function of the overlap group-Lasso in equation \eqref{eq:min3_init} can be decomposed into the sum of three terms:
\begin{equation}
\ell(\boldsymbol{\psi})=\frac{1}{2}\Vert \mathbf{y} - \mathbf{Z} \boldsymbol{\psi}\Vert_2^2 + \lambda\sum_{b=1}^{B+1}\sqrt{\boldsymbol{\psi}^T\mathbf{D}_b^2\boldsymbol{\psi}}=\mathrm{RSS}(\boldsymbol{\psi})+\mathcal{P}^{\mathrm{OG}}_\lambda(\boldsymbol{\psi}),
\end{equation}
having subdifferential 
\begin{equation}
\partial\ell(\boldsymbol{\psi})=\Bigg\{-\mathbf{Z}^T(\mathbf{y}-\mathbf{Z}\boldsymbol{\psi})+\mathcal{P}^{\mathrm{OG}}_\lambda(\boldsymbol{\psi})\Bigg\},
\end{equation}
where $\partial\mathcal{P}^{\mathrm{OG}}_\lambda(\boldsymbol{\psi})=\partial\big(\lambda\sum_{b=1}^{B+1}\sqrt{\boldsymbol{\psi}^T\mathbf{D}_b^2\boldsymbol{\psi}}\big)$ is the subdifferential of the OG penalty term
\begin{equation}
\begin{aligned}
\partial\mathcal{P}^{\mathrm{OG}}_\lambda(\boldsymbol{\psi})&=\partial\big(\lambda\sum_{b=1}^{B+1}\sqrt{\boldsymbol{\psi}^T\mathbf{D}_b^2\boldsymbol{\psi}}\big)\\
&=\Big\{
\mathbf{s}\in\mathbb{R}^{LM}\,:\,\mathbf{s}=\sum_{b=1}^{B+1}\frac{\mathbf{D}_b^2\boldsymbol{\psi}}{\sqrt{\boldsymbol{\psi}^T\mathbf{D}^2_b\boldsymbol{\psi}}},\,\mathrm{if}\,\Vert\boldsymbol{\psi}\Vert_2>0,\,\mathrm{and}\,\Vert\boldsymbol{\psi}\Vert_2\leq\lambda\sum_{b=1}^{B+1}\mathbf{D}_b^2,\,\mathrm{otherwise}
\Big\}.
\end{aligned}
\end{equation}
Since the objective function is convex the first order optimality condition $\partial\ell(\boldsymbol{\psi})=0$ and Theorem \ref{th:existence} ensure the existence of a unique $\widehat{\mathbf{s}}\in\partial\mathcal{P}^{\mathrm{OG}}_\lambda(\boldsymbol{\psi})$ fulfilling the condition $-\mathbf{Z}^T(\mathbf{y}-\mathbf{Z}\widehat{\boldsymbol{\psi}})+\widehat{\mathbf{s}}=0$. Plugging $\widehat{\boldsymbol{\psi}}=0$ we get $\mathbf{Z}^T\mathbf{y}=\lambda\widehat{\mathbf{s}}$ from which we obtain the equivalent condition $\big(\sum_{b=1}^{B+1}\mathbf{D}_b^2\big)^{-1}\mathbf{Z}^T\mathbf{y}=\lambda\widehat{\mathbf{z}}$ and $\widehat{\mathbf{z}}$ is the dual variable satisfying $\widehat{z}_j=\mathrm{sign}(\widehat{\psi}_j)$ if $\widehat{\psi}_j\neq 0$ and $\widehat{z}_j\in[-1,1]$ if $\widehat{\psi}_j= 0$ for $j=1,\dots,LM$. Therefore $\Big\Vert\big(\sum_{b=1}^{B+1}\mathbf{D}_b^2\big)^{-1}\mathbf{Z}^T\mathbf{y}\Big\Vert_\infty=\lambda\Vert\widehat{\mathbf{z}}\Vert_\infty$, which completes the proof.
%
\end{proof}

\subsection*{Algorithm convergence properties}
%
\noindent The following propositions state the local and global convergence properties of the MM algorithm. In what follows, let $\ell(\mathbf{\boldsymbol{\psi}})$ be the objective function to be minimized and $\mathcal{Q}(\mathbf{\boldsymbol{\psi}}\vert\mathbf{\widehat{\boldsymbol{\psi}}}^{(k)})$ be the majorizing function at the current $k$-iteration, where $\widehat{\boldsymbol{\psi}}^{(k)}$ is the estimate of the parameter $\boldsymbol{\psi}$ at the $k$-iteration. Moreover, let $\widehat{\boldsymbol{\psi}}^{(k+1)}=M(\widehat{\boldsymbol{\psi}}^{(k)})$ denote the mimimizer of $\mathcal{Q}(\mathbf{\boldsymbol{\psi}}\vert\widehat{\mathbf{\boldsymbol{\psi}}}^{(k)})$ at iteration $k$. Following \cite{lange.2010} we provide conditions for local and global convergence of the MM algorithm for proving the local and global convergence of the MM sequence delivered by the MM algorithm introduced in Section \ref{sec:comp}.\newline

\noindent The following proposition \citep[adapted from][Proposition 15.3.2]{lange.2010} provides the conditions for local convergence of any MM sequence.
\begin{proposition}
Since the majorizing function is strictly convex, if the Hessian of the surrogate function $d^{20}\mathcal{Q}(\mathbf{\boldsymbol{\psi}}^{\infty}\vert\mathbf{\boldsymbol{\psi}}^{\infty})$ is invertible, then the proposed MM algorithm is locally attracted to a local minimum $\mathbf{\boldsymbol{\psi}}^{\infty}$ at a linear rate equal to the spectral radius $\rho(dM(\boldsymbol{\psi}^{\infty}))$ of $dM(\boldsymbol{\psi}^{\infty})=\mathbf{I}_{LM}-\big(d^{20}\mathcal{Q}(\boldsymbol{\psi}^{\infty}\vert\mathbf{\boldsymbol{\psi}}^{\infty})\big)^{-1}d^2\ell(\mathbf{\boldsymbol{\psi}}^{\infty})$, where $\mathbf{I}_{LM}$ denotes the identity matrix of dimension $LM$.
\label{prop:spectralradius}
\end{proposition}
\begin{proof}
The mapping function in equation \eqref{eq:minimization_surrogate_fun_sol} is differentiable and the surrogate function is strictly convex. Moreover, the second derivative of $\mathcal{Q}(\mathbf{\boldsymbol{\psi}}^{\infty}\vert\mathbf{\boldsymbol{\psi}}^{\infty})$ is invertible for any $\lambda>0$. By Proposition 15.3.1 in \cite{lange.2010} it is then sufficient to show that all the eigenvalues of the differential $dM(\boldsymbol{\psi}^{\infty})$ lie in the half-open interval $[0,1)$. All the eigenvalues of the $dM(\boldsymbol{\psi}^{\infty})$ can be determined as stationary points of the Rayleight quotient $\mathcal{R}_\upsilon(\widehat{\boldsymbol{\psi}}^{(\infty)})=\frac{\upsilon^T(d^{20}\mathcal{Q}(\boldsymbol{\psi}^{\infty}\vert\mathbf{\boldsymbol{\psi}}^{\infty})-d^2\ell(\mathbf{\boldsymbol{\psi}}^{\infty}))\upsilon}{\upsilon^T(d^{20}\mathcal{Q}(\boldsymbol{\psi}^{\infty}\vert\mathbf{\boldsymbol{\psi}}^{\infty})\upsilon}=1-\frac{\upsilon^Td^2\ell(\mathbf{\boldsymbol{\psi}}^{\infty})\upsilon}{\upsilon^T(d^{20}\mathcal{Q}(\boldsymbol{\psi}^{\infty}\vert\mathbf{\boldsymbol{\psi}}^{\infty})\upsilon}$. Moreover, strict convexity of both $\ell(\mathbf{\boldsymbol{\psi}})$ and $\mathcal{Q}(\mathbf{\boldsymbol{\psi}}\vert\widehat{\mathbf{\boldsymbol{\psi}}}^{(k)})$ implies that $d^2\ell(\mathbf{\boldsymbol{\psi}}^{\infty})$ and $d^{20}\mathcal{Q}(\mathbf{\boldsymbol{\psi}}^\infty\vert\mathbf{\boldsymbol{\psi}}^{\infty})$ are positive definite \citep[see][]{polak.1987} and therefore $\mathcal{R}_\upsilon(\widehat{\boldsymbol{\psi}}^{(\infty)})<1$ for any non null vector of length one. Positive semi-definiteness of $d^{20}\mathcal{Q}(\mathbf{\boldsymbol{\psi}}^\infty\vert\mathbf{\boldsymbol{\psi}}^{\infty})-d^2\ell(\mathbf{\boldsymbol{\psi}}^{\infty})$ also implies $\mathcal{R}_\upsilon(\widehat{\boldsymbol{\psi}}^{(\infty)})\geq0$.
\end{proof}
The convergence rate is usually used to characterize the convergence behavior of an iterative algorithm. It is well known that the convergence rate of an MM  algorithm is, in general, linear. If $\{\widehat{\boldsymbol{\psi}}^{(k)}\}$ converges to some optimal point $\boldsymbol{\psi}^\infty$ of $\ell(\boldsymbol{\psi})$ and $M(\boldsymbol{\psi})$ is
continuous, then $\boldsymbol{\psi}^\infty$ is a fixed point and $\boldsymbol{\psi}^\infty=M(\boldsymbol{\psi}^\infty)$. By Taylor expansion, $\widehat{\boldsymbol{\psi}}^{(k+1)}-\boldsymbol{\psi}^{\infty}=dM(\boldsymbol{\psi}^{\infty})(\widehat{\boldsymbol{\psi}}^{(k)}-\boldsymbol{\psi}^{\infty})$ where $dM(\boldsymbol{\psi}^{\infty})$ is the matrix rate of convergence. The spectral radius of $dM(\boldsymbol{\psi}^{\infty})$ is usually defined as the local convergence rate of the sequence \citep[see, e.g.][]{lange.2010,mclachlan_etal.2008}.
The following proposition provides conditions for global convergence of the MM algorithm.
\begin{proposition}
If $\ell(\boldsymbol{\psi})$ is coercive, the subset $\{\boldsymbol{\psi}\in \boldsymbol{\Theta}\backepsilon\ell(\boldsymbol{\psi})\leq\ell(\widehat{\boldsymbol{\psi}}^{(k)})\}$ of parameter domain $\boldsymbol{\Theta}$ is compact and all stationary points of $\ell(\boldsymbol{\psi})$ are isolated. The majorizing function $\mathcal{Q}(\mathbf{\boldsymbol{\psi}}\vert\widehat{\mathbf{\boldsymbol{\psi}}}^{(k)})$ is strictly convex and differentiable in both $\boldsymbol{\psi}$ and $\widehat{\boldsymbol{\psi}}^{(k)}$, then the MM sequence $\{\widehat{\boldsymbol{\psi}}^{(k)}\}$ convergence to the stationary point of $\ell(\boldsymbol{\psi})$. Moreover, since $\ell(\boldsymbol{\psi})$ is strictly convex, then the limiting point of $\{\widehat{\boldsymbol{\psi}}^{(k)}\}$ is the minimum.
\label{prop:convergalg_global}
\end{proposition}
\begin{proof}
For any fixed $\lambda>0$, the objective function $\ell(\boldsymbol{\psi})$ is convex with one bounded local minimizer (see Proposition \ref{th:existence}), which implies that $\ell(\boldsymbol{\psi})$ is coercive, \citep[see][]{lange.2010,lange.2016}. The convexity of $\ell(\boldsymbol{\psi})$ further implies that $\ell(\boldsymbol{\psi})$ is Lipschitz continuous on each compact subset of $\mathbb{R}^{LM}$ (i.e. locally Lipschitz continuous). It follows that the gradient of $\ell(\boldsymbol{\psi})$ exists for almost all $\boldsymbol{\psi}$. Moreover, by the Liapunov theorem \citep[Proposition 15.4.1]{lange.2010} the set $\mathcal{D}$ of clustering points generated by the sequence $\{\widehat{\boldsymbol{\psi}}^{(k)}\}$, for $k=0,1,\dots$, is contained in the set $\mathcal{S}$ of stationary points of $\ell(\boldsymbol{\psi})$. Then, according to \citep[Proposition 8.2.1]{lange.2010}, $\mathcal{D}$ is a connected set since any closed subset of a compact set is also compact. The condition that all stationary points of $\ell(\boldsymbol{\psi})$ are isolated easily implies that the number of stationary points in the compact set $\{\boldsymbol{\psi}\in \boldsymbol{\Psi}\backepsilon\ell(\boldsymbol{\psi})\leq\ell(\widehat{\boldsymbol{\psi}}^{(k)})\}$ can only be finite and since the cluster set $\mathcal{D}$ is a connected subset of a finite set $\mathcal{S}$, $\mathcal{D}$ reduces to a singleton.
\end{proof}

\subsection*{Proofs of Section 4}

\begin{proof}[Remark	\ref{rem:concurrentseidentro}]
\[
\mu( \bar \F_0 \backslash \tilde \F_0^\text{concurrent}) = 
\mu( \bar \F_0) = 
(M-d+2) \frac{c}{(M-d+2)^2 } = c(M-d+2)^{-1}.
\]
\end{proof}      
	
	\begin{proof}[Theorem \ref{th:existence}]
		Using Proposition 1  of \citet{jenatton2011structured} and noting that the $B$-th element of the second sum in \eqref{eq:min3_init} is the norm of all the coefficients $\psi_{ml}$, the minimization in \eqref{eq:min3_init} leads to a unique set of coefficients $\widehat{\psi}_{ml}$. 
		Since $d$ and the knots defining the bases $\{\varphi_m(t), m=1,2,\dots,M\}$ and $\{\theta_l(s), l=1,2,\dots,L\}$ are fixed each B-spline basis represents a vector space of piecewise polynomials. Since their tensor product also generates a vector space, the elements of the tensor product basis are also linearly independent and so the coefficients $\psi_{ml}$ of any given kernel $\psi(t,s)$ in this vector space are uniquely determined. Thus, the estimator  $\widehat \psi_\lambda$ obtained by the minimization of \eqref{eq:min3_init} is unique. 
	\end{proof}
	
	\begin{proof}[Theorem \ref{th:consistency1}]
		We first note that if $\psi_T \in V_{\varphi \otimes \theta}$, there exists a unique true $\Psib_T$ and thus  the event $E_0$ is equivalent to consistently estimate the non-zero patterns in $\Psib$. To show this, it is sufficient to rely on  Theorem 6 of \citet{jenatton2011structured}. The rest of the proof consists in noting  that the conditions required by  Theorem 6 of \citet{jenatton2011structured} are satisfied under our set of assumptions. 
	\end{proof}

	\begin{proof}[Theorem \ref{th:consistency2}]
		We study the consistency of $\widehat{\psi}_\lambda(t,s)$ with respect to $\psi_T^P(t,s)$ exploiting its representation through the vector space $V_{\varphi \otimes \theta}$ and using the results of Theorem 7 of \citet{jenatton2011structured}. Eventually the consistency of $\widehat{\psi}_\lambda(t,s)$ with respect to the true $\psi_T(t,s)$ is obtained letting the number of B-splines basis to grow with $n$ since if both $L$ and $M$ are increasing we have that the event $E_k$ approaches $E_0$ and $\int [\psi_T^P(t,s) - \psi_T(t,s)]^2 dt ds \to 0$.
		%
		Note that 
		for any ${\F_1 \subset  \F_{0T}^C} $ we also have ${\F_1 \subset \bar{\F}^C_{0T}} $ since $\bar \F_{0T}$ is a subset of $\F_{0T}$. Thus the  correct estimation of $\Psib_T^P$  automatically leads to 
		\[
		\left(\inf_{\F_1 \subset \F_{0T}^C} 
		\int_{\F_{1}} \left[\widehat{\psi}_\lambda (t,s)\right]^2 dt ds\right) >0.
		\]
		Despite correctly estimating $\Psib_T^P$, however, it may happen that $\widehat{\psi}_\lambda(t,s)\neq 0$ for $(t,s) \in \F_{0T}$ for subregions of a finite number of sets $\F_{ml}$ in \eqref{eq:nullregion} of Lebesgue measure $(LM)^{-1}$. Since $\widehat{\psi}_\lambda(t,s)$ is bounded  we can  conclude that a correct estimation of the non zero patterns in $\Psib_T^P$ would also lead to  
		\[
		\int_{\F_{0T}} \left[ \widehat{\psi}_\lambda (t,s)\right]^2 dt ds \leq \frac{k}{LM}.
		\]
		Hence, similarly to  Theorem \ref{th:consistencyfixeddim} we focus on the consistent estimation of  the non-zero patterns in $\Psib^P_T$. To this end, we exploit Theorem 7 of \citet{jenatton2011structured} and show how the conditions required therein translate to our settings.
		We first sightly rewrite the minimization in \eqref{eq:min3_init} as
		\begin{equation}
		\widehat \bpsi_\lambda = \arg \min_{\bpsi} \left\{\frac{1}{2} || \mathbf{y} - {\tilde{\mathbf{Z}}} \bpsi||_2^2 + \lambda \sum_{b=1}^{B+1}  ||  \tilde{c}_{b} \odot \bpsi ||_2 \right\} \odot \mbox{V}_{{\mathbf{Z}}}^{-1},
		\label{eq:min4}
		\end{equation}
		where ${\tilde{\mathbf{Z}}}$ is obtained dividing each column of ${\mathbf{Z}}$ by its standard deviation. Consistently with this operation we are intrinsically rescaling $\bpsi$ and thus the final estimator is multiplied through the Hadamart product for the inverse of $\mbox{V}_{{\mathbf{Z}}}$ which is the vector containing the standard deviations of  ${\mathbf{Z}}$.
		To maintain the penalization consistent with our settings, the weights inside the norm are also redefined as $\tilde c_b = c_b \odot \mbox{V}_{\mathbf{Z}}^{-1}$.
		With this operation, the minimization is conducted on a matrix ${\tilde{\mathbf{Z}}}$ such that ${\tilde{\mathbf{Z}}}^{T}{\tilde{\mathbf{Z}}}$ has unit diagonal. 
		In addition, let $\kappa$ be the minimum eigenvalue of the matrix 
		${\tilde{\mathbf{Z}_1}}^{T}{\tilde{\mathbf{Z}}_1}$. 	  From condition \eqref{eq:sparsitycondition}  the number of nonzero coefficients in $\Psib_T^P$ is strictly less than $nG$ and then 
		${\tilde{\mathbf{Z}_1}}^{T}{\tilde{\mathbf{Z}_1}}$
		is positive definite and $\kappa >0$. 
		
		Hence, defining the constants  $C_1, C_2, C_3,$ and $C_4$ consistently with \citet{jenatton2011structured} all the requirements of their Theorem 7 are satisfied. Hence the proof.
	\end{proof}